\documentclass[10pt]{article}
\usepackage[paperwidth=8.5in,paperheight=11.00in,top=1.25in, bottom=1.25in, left=1.00in, right=1.00in]{geometry}
\usepackage{microtype}

\usepackage{amsmath}
\usepackage{amssymb}
\usepackage{mathtools}          
\usepackage{physics}
\usepackage{unicode-math}
\mathtoolsset{showonlyrefs=true}       
\numberwithin{equation}{section}
\allowdisplaybreaks 

\usepackage{graphicx}
\usepackage{subcaption}
\usepackage{booktabs}

\usepackage{amsthm}
\theoremstyle{plain}
\newtheorem{theorem}{Theorem}[section]
\newtheorem{lemma}[theorem]{Lemma}
\newtheorem{proposition}[theorem]{Proposition}
\newtheorem{corollary}[theorem]{Corollary}
\theoremstyle{definition}

\newtheorem{remark}[theorem]{Remark}

\usepackage[authoryear]{natbib}
\usepackage[dvipsnames]{xcolor}
\usepackage[colorlinks=true, citecolor=RoyalBlue, linkcolor=RubineRed, urlcolor=Turquoise, pdfpagelabels
]{hyperref}

\newcommand{\one}{\mathbb{1}}
\newcommand{\E}{\mathbb{E}}

\newcommand{\N}{\mathbb{N}}

\newcommand{\R}{\mathbb{R}}

\newcommand{\cD}{\mathcal{D}}

\newcommand{\cK}{\mathcal{K}}

\newcommand{\cT}{\mathcal{T}}

\newcommand{\cU}{\mathcal{U}}

\newcommand{\what}{\widehat}
\newcommand{\wtilde}{\widetilde}

\makeatletter
\def\wbreve{\mathpalette\wide@breve}
\def\wide@breve#1#2{\sbox\z@{$#1#2$}%
     \mathop{\vbox{\m@th\ialign{##\crcr
\kern0.08em\brevefill#1{0.8\wd\z@}\crcr\noalign{\nointerlineskip}%
                    $\hss#1#2\hss$\crcr}}}\limits}
\def\brevefill#1#2{$\m@th\sbox\tw@{$#1($}%
  \hss\resizebox{#2}{\wd\tw@}{\rotatebox[origin=c]{90}{\upshape(}}\hss$}
\makeatother

\DeclareMathOperator{\sign}{sign}
\DeclareMathOperator*{\argmin}{arg\,min}

\begin{document}
\title{One Other Option Pricing Scheme}
\author{
    Jimin Lin\thanks{Quantitative Research, Bloomberg. (E-mail: \href{mailto:jlin846@bloomberg.net}{jlin846@bloomberg.net}).}
}
\date{}
\maketitle

\begin{abstract}
We present a distinctive approach to parameterizing the risk neutral distribution. Using parsimonious and interpretable parameters, the model provides direct and localized control over the shape of the implied volatility curve. It captures a wide variety of shapes, including those with local concavity. Empirical results demonstrate accurate calibration across a quarter million curves from a two-year Standard and Poor's 500 index option dataset. The fitted parameters exhibit stable patterns across tenors, enabling term structure interpolation and dynamic process construction without static arbitrage.
\end{abstract}

\section{Introduction}
Evidence from short-dated option markets suggests a broader range of implied volatility curve shapes than those typically emphasized in earlier studies. Specifically, implied volatility curves intermittently depart from the classical convex shape, retrospectively recognized as the U-shape, and exhibit local concavity or multiple extrema, often associated with scheduled events \citep{baker2018trading, klassen2023state, glasserman2023w, alexiou2025pricing}. By concavity location and magnitude, curves are classified as W-shaped with central concavity, inverse U-shaped with strong central concavity, or S-shaped with lateral concavity.

These variations, more frequently observed in modern option markets, are more intricate than what classical models were designed to capture, motivating renewed interest in innovating or renovating models for greater expressivity. To clarify the design space and position our proposed model, we review existing and applicable approaches in detail. Broadly, they can be grouped into three paradigms.

The first line of work extends classical stochastic differential equation models with ad hoc factors. \citep{alexiou2025pricing} adds two extra anticipated jumps to the Bates stochastic volatility jump model to capture event risk. This model calibrates to the concave implied volatility on the event date and reduces to the baseline model on ordinary dates. \citep{zhou20250dte} applies a more concise Bachelier model with jumps to capture the concavity, and insightfully identifies the resulting marginal distribution as a specific Gaussian mixture with Poisson weights. This bridges to the second type of modeling.

The second approach enhances the expressivity of existing models by compounding distributions. By exploiting the universal approximation property of Gaussian mixtures, the model of \citep{brigo2002lognormal} naturally captures concavity by increasing the number of components. The mechanism for the Gaussian mixture model to produce curve concavity is explained by \citep{glasserman2023w}. The variance gamma mixture model investigated by \citep{keller2025w} is a comparable alternative that addresses the flattening wings of the Gaussian mixture model. \citep{zaugg2025volatility} places these mixture models into a more general framework by randomizing parameters of classical models.

The third methodology directly parametrizes implied volatility surface with constraints to rule out static arbitrage. The parameterization in \citep{itkin2015sigmoid} leverages the universal approximation capability of sigmoid functions and may capture concavity. The proprietary model in \citep{klassen2023state} likely functions with higher order expansions and specific terms. The use of neural networks, such as \citep{ackerer2020deep, zheng2021incorporating, wiedemann2024operator, yang2025hyperiv, lin2026shallow}, can be viewed as overparameterization. \citep{keller2026discovering} applies symbolic regression to search for new candidates.

Taken together, these approaches reflect a fundamental tension between parsimony and expressivity and between interpretability and complexity. Parsimonious models offer stability and interpretability but may restrict the range of attainable shapes, while more flexible constructions accommodate richer empirical features at the cost of additional parameters or structural complexity. Notably, all aforementioned workable models, regardless of their lineage, ultimately rely, either directly or indirectly, on universal approximators, such as compound distributions, high order expansions, or activation function superpositions, with the balance often favoring expressivity over parsimony and complexity over interpretability.

We introduce a distinct option pricing scheme that strikes a balance between these objectives. By directly parameterizing the implied risk neutral quantile function, we create a family of models that retains parsimonious and intuitive structure while remaining sufficiently expressive to reproduce the range of implied volatility shapes observed in practice. The parameters explicitly identify moneyness regimes and characterize the behavior of implied volatility and implied density within each regime. Empirical results based on a quarter million implied volatility curves over the past two years show that the proposed framework achieves accurate and stable calibration across all tenors, moneynesses, and market conditions.

The rest of the paper is structured as follows. In the remainder of this section, we explain the model conventions and some preliminaries. Section~\ref{sec:quantile} illustrates the use of quantile functions in the option pricing setting, where we progressively build towards a practically useful model. The model is then formalized in Section~\ref{sec:seqs}. Section~\ref{sec:empirical} applies the model to a large dataset to empirically evaluate its performance. An optional way to extrapolate the model family is discussed in Section~\ref{sec:affine}. Section~\ref{sec:dynamic} discusses an approach to construct a dynamic process with the feasible term structure of calibrated parameters.

\subsection{Notation}
We begin in a static setting that presumes a fixed date and expiration, and thus omit time notation until it becomes necessary. Let $X$ be the \emph{risk neutral (logarithmic) return} such that $\E\qty[e^X] = 1$, with cumulative distribution function $F$ and probability density function $f$ respectively referred to as the \emph{risk neutral distribution} and the \emph{risk neutral density}. For simplicity, cumulative distribution functions in this context are all assumed to be continuous and strictly increasing. The corresponding quantile function $Q = F^{-1}$ and quantile density function $q = Q'$ are the \emph{risk neutral quantile} and the \emph{risk neutral quantile density}. Denote the complementary cumulative distribution function by $\overline{F} := 1 - F$. Define the tilted density and the tilted distribution by $\wtilde{f}(x):= e^xf(x)$  and $\wtilde{F}(x)= \int_{-∞}^x \wtilde{f} \dd x$.

The prices of normalized European style \emph{put} $p(κ):=\E\qty[\qty(e^κ - e^X)^+]$ and \emph{call} $c(κ):=\E\qty[\qty(e^X - e^κ)^+]$ quoted in \emph{(logarithmic and forward-adjusted) moneyness} $κ$ are given by
\begin{equation} \label{eq:price}
\begin{aligned}
p(κ)
    &= e^κ F(κ) - \wtilde{F}(κ)
    = e^κ Q^{-1}(κ) - ∫_0^{Q^{-1}(κ)} e^{Q(u)} \dd u, \\
c(κ)
    &=\overline{\wtilde{F}}(κ) - e^κ \overline{F}(κ)
    = ∫_{Q^{-1}(κ)}^1 e^{Q(u)} \dd u - e^κ \overline{Q^{-1}}(κ),
\end{aligned}    
\end{equation}
where the expressions in terms of $Q$ follow from a change of variables. The \emph{implied (total) volatility} $ω$ is the value that matches the option price given by the Black-Scholes pricing formula to the option prices~\eqref{eq:price} with
\begin{equation} \label{eq:price_bs}
\begin{aligned}
p(κ) &= e^κ Φ\qty(z^+(κ, ω(κ))) - Φ\qty(z^-(κ, ω(κ))), \\
c(κ) &= Φ\qty(-z^-(κ, ω(κ))) - e^κ Φ\qty(-z^+(κ, ω(κ))),
\end{aligned}
\end{equation}
where $Φ$ is the standard normal cumulative distribution function and $z^±$ are pivots of the risk neutral normal distribution and its tilted distribution given by
\begin{align}
z^±(x, v) := \frac{x}{v} ± \frac{v}{2}.
\end{align}

\subsection{Preliminary}
The known relations between the implied distribution and implied volatility given below, albeit not always directly applicable, offer valuable heuristics for constructing the risk neutral quantile. Proposition~\ref{prop:rnd_iv} provides exact expressions of implied distribution and implied density in terms of implied volatility, which have been derived in different forms for various purposes \citep{jackwerth2000recovering, brunner2003arbitrage, durrleman2010implied, roper2009implied, tavin2012implied}; see \citep{lin2026shallow} for a review. Proposition~\ref{prop:asym} describes the implied volatility asymptotics, which is adapted from the moment formula in \citep[Theorem 3.2 and 3.4]{lee2004moment} and wing formula in \citep[Theorem 1.1 and 1.2]{benaim2009regular}.

\begin{proposition} [Implied quantity relation] \label{prop:rnd_iv}
Let $F$ be a risk neutral distribution, $f$ be its density, and $ω$ be its corresponding implied volatility that is assumed to be twice differentiable with respect to moneyness $κ$. Then
\begin{equation} \label{eq:rnd_iv}
\begin{aligned}
F(κ)
    &= Φ(z^+) + ϕ(z^+) ω', \\
f(κ)
    &= \frac{ϕ(z^+)}{ω} \qty(1 + z^+z^- (ω')^2 - \qty(z^+ + z^-) ω' + ω ω''),
\end{aligned}    
\end{equation}
where
\begin{align}
z^± := z^±(κ, ω(κ)), &&
ω := ω(κ), &&
ω' := \frac{\dd ω(κ)}{\dd κ}, &&
ω'' := \frac{\dd^2 ω(κ)}{\dd κ^2}.
\end{align}
\end{proposition}

We introduce a few more notations to describe implied volatility asymptotics. Denote the tail of cumulative distribution $F$ by
\begin{equation}
T(x) = \begin{cases}
\overline{F}(x), & x >0, \\
F(x), & x ≤0.
\end{cases}    
\end{equation}
Define the function $ψ: \R → [0, 2]$ by
\begin{align}
ψ(x) := 2 - 4\qty(\sqrt{x^2 + x} - x).
\end{align}
We write $g(κ) ∼ h(κ)$ to mean $g(κ)/h(κ) → 1$ as either $κ → -∞$ or $κ → +∞$, unless the direction is specified.

\begin{proposition}[Implied volatility asymptotic] \label{prop:asym}
Under the same assumptions as in Proposition~\ref{prop:rnd_iv}, define
\begin{align} \label{eq:moment}
ε^+ := \sup\qty{ε: \E \qty[e^{(1+ε)X}] < ∞}, &&
ε^- := \sup\qty{ε: \E \qty[e^{-ε X}] < ∞}.
\end{align}
The following hold:
\begin{enumerate}
    \item Moment formula:
        \begin{equation} \label{eq:iv_asym_bound}
            \limsup \frac{ω^2(κ)}{\abs{κ}} =
            \begin{cases}
                ψ(ε^+), & κ → +∞. \\
                ψ(ε^-), & κ → -∞.
            \end{cases}
        \end{equation}
    \item Wing formula: if additionally $ε^+, ε^- > 0$, then
        \begin{equation} \label{eq:iv_asym_lin}
            \frac{ω^2(κ)}{\abs{κ}} ∼ ψ\qty(-\frac{\log e^κ T(κ)}{\abs{κ}}).
        \end{equation}
        If furthermore $-\frac{\log T(κ)}{\abs{κ}} → ∞$ as $κ → ± ∞$, then
        \begin{equation} \label{eq:iv_asym_sub}
            \frac{ω^2(κ)}{\abs{κ}} ∼ -\frac{\abs{κ}}{\log T(κ)},
        \end{equation}
\end{enumerate}
\end{proposition}

\section{Pricing options with quantile functions} \label{sec:quantile}
While many other ideas from probability have been exhaustively exploited in option pricing theory, the quantile function remains vastly underutilized. In contrast to cumulative distribution functions or probability density functions, which live in nonlinear convex sets subject to boundary or normalization constraints, the set of quantile functions forms a convex cone defined only by a monotone constraint.

In other words, the validity of a quantile function is preserved under a wider range of operations, such as affine transformation, addition, monotone transformation, and certain composition, that cannot be applied to the other two. This flexibility is particularly useful for constructing new distributions to accommodate the increasingly intricate shapes of risk neutral distributions observed in modern option markets. General treatments of the quantile function can be found in \citep{gilchrist2000statistical}, and a universal coordinate system for quantiles is introduced by \citep{keelin2016metalog}.

In this section, we warm up with certain operations on the quantile function by constructing four motivating option pricing models. Through conic combination, composition, monotone transformation, and segmental construction, we obtain, respectively, a risk neutral asymmetric logistic distribution, a beta logistic distribution that relates to \citep{carr2021additive, azzone2025explicit}, a stretched logistic distribution, and a new family of distributions. These constructions progress toward the model in Section~\ref{sec:seqs}.

A risk neutral quantile can be obtained from a raw quantile with an additive shift, analogous to the multiplicative Esscher exponential tilting.
\begin{proposition}[Quantile risk neutralization] \label{prop:neutr}
Let $\wbreve{Q}$ be a quantile function and $\wbreve{X} ∼ \wbreve{Q}^{-1}$. If $\E \qty[e^{\wbreve{X}}] < ∞$, then the shifted quantile $Q$ below is risk neutral:
\begin{align} \label{eq:neutr}
Q(u) = μ + \wbreve{Q}(u), &&
μ = - \log ∫_0^1 e^{\wbreve{Q}(u)} \dd u.
\end{align}
\end{proposition}
\begin{proof}
For $X ∼ Q^{-1}$, $\E\qty[e^X] = ∫_{\R} e^x \dd F(x) = ∫_0^1 e^{μ + \wbreve{Q}(u)} \dd u = \qty(∫_0^1 e^{\wbreve{Q}(u)} \dd u)^{-1}∫_0^1 e^{\wbreve{Q}(u)} \dd u = 1$.
\end{proof}

\subsection{Conical combination: From exponential to logistic} \label{sec:logistic}
\begin{figure}[h]
    \centering
    \includegraphics[width=0.95\textwidth]{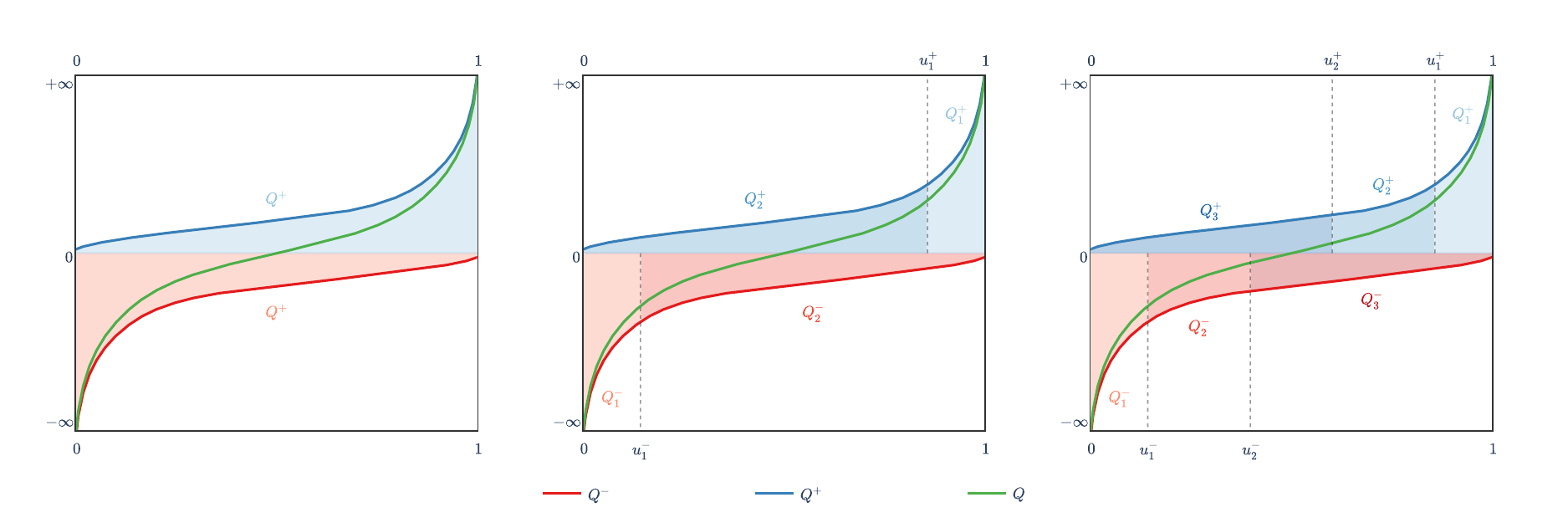}
    \caption{Segmentally constructing quantile functions.}
    \label{fig:qf}
\end{figure}

A conical combination of quantile functions yields another valid quantile function. Denote by $Q^+$ the quantile function of exponential distribution with unit mean and by $Q^-$ the quantile function of negative exponential distribution with negative unit mean:
\begin{align} \label{eq:qf_exp}
Q^+(u) = -\log(1-u), &&
Q^-(u) =  \log u.
\end{align}
Scaling $Q^±$ respectively by $ς^±$ preserves the exponential tails, whereas the conic combination
\begin{equation} \label{eq:qf_l}
Q_L(u; ς^+, ς^-) := ς^- Q^-(u) + ς^+ Q^+(u)  = ς^- \log u - ς^+ \log(1-u),
\end{equation}
becomes the quantile function of asymmetric logistic distribution. When $ς^- = ς^+$, it reduces to the familiar logistic distribution; otherwise, it has no explicit cumulative distribution function. Figure~\ref{fig:qf} (left) illustrates such a conical combination.

\subsubsection*{Risk neutralization}
Let $\wbreve{Q}(u) = Q_L(u; ς^+, ς^-)$ and $\wbreve{X} ∼ \wbreve{Q}^{-1}$. To apply Proposition~\ref{prop:neutr}, we evaluate the moment generating function
\begin{align}
\E\qty[e^{λ\wbreve{X}}] = ∫_0^1 e^{λ\wbreve{Q}(u)} \dd u = ∫_0^1 u^{λς^-} (1-u)^{-λς^+} \dd u = B(1 + λς^-, 1 - λς^+),
\end{align}
where $B$ is the beta function. It implies the integrability condition
\begin{equation}
 \E \qty[e^{\wbreve{X}}] < ∞ \text{ if and only if } ς^+ < 1.
\end{equation}
Thus, the risk neutral asymmetric logistic quantile is
\begin{align}
Q(u) = μ + \wbreve{Q}(u), &&
μ = -\log B(1 + ς^-, 1-ς^+), &&
\wbreve{Q}(u) = ς^- \log u - ς^+ \log (1-u),
\end{align}

\subsubsection*{Pricing formula}
The option price is obtained from formula~\eqref{eq:price} by substituting
\begin{align}
F(x) = Q^{-1}(x), && \wtilde{F}(x) = ∫_0^{F(x)} e^{Q(u)} \dd u = F_B\qty\big(F(x); 1 + ς^-, 1 - ς^+),
\end{align}
where 
\begin{equation} \label{eq:cdf_b}
F_B(x; α, β) = \frac{B(x; α, β)}{B(α, β)}    
\end{equation}
is the beta cumulative distribution function with parameters $α$ and $β$, and $B(x; α, β)$ is the incomplete beta integral. In the symmetric case $ς^± = ς$, $F$ has an explicit form given by the standard logistic distribution $F_L$ with a pivot $z$ as
\begin{align} \label{eq:cdf_l}
F(x) = F_L\qty\big(z(x; μ, ς)), &&
F_L(z) = \frac{e^z}{1 + e^z}, &&
z(x; μ, ς) = \frac{x - μ}{ς}.
\end{align}

\subsubsection*{Wing asymptotic} 
$Q$ preserves the tail behaviors of $Q^-$ and $Q^+$, with $Q^-$ dominating as $u → 0$ and $Q^+$ dominating as $u → 1$. Thus,
\begin{align}
T(x) ∼ 
\begin{cases}
    e^{-\frac{x}{ς^+}}, & x → +∞ \\ 
    e^{\frac{x}{ς^-}}, &  x → -∞,
\end{cases}
&&
\lim -\frac{\log T(x)}{\abs{x}} =
\begin{cases}
    \frac{1}{ς^+}, & x → +∞, \\
    \frac{1}{ς^-}, & x → -∞, \\
\end{cases}
\end{align}
By formula~\eqref{eq:iv_asym_lin} in Proposition~\ref{prop:asym}, the wing asymptotics are
\begin{equation}
\frac{ω^2(κ)}{\abs{κ}} ∼
\begin{cases}
ψ\qty(\frac{1}{ς^+} - 1), & κ → +∞, \\
ψ\qty(\frac{1}{ς^-} ), & κ → -∞.
\end{cases}
\end{equation}
The right tail integrability condition $ς^+ < 1$ for the first line inherits from risk neutrality, while the left tail condition $ς^- < 1$ is additionally required for the second line. 

\begin{remark}
Logistic distribution is an edge case of feasible risk neutral distribution. It achieves linear asymptotic wing slope, the theoretical upper bound~\eqref{eq:iv_asym_bound}, with the moment condition~\eqref{eq:moment} determined by the scale parameter $ς$.
\end{remark}

\subsection{Composition: From logistic to beta logistic} \label{sec:betalogistic}
The composition of a quantile function with another special quantile function with unit codomain is still a quantile function. We compose the logistic quantile~\eqref{eq:qf_l} with the beta quantile~\eqref{eq:cdf_b},
\begin{align}
Q_L(u) := Q_L(u; ς^-, ς^+), &&
Q_B(u) := Q_B(u; α, β) = F_B^{-1}(u; α, β),
\end{align}
into
\begin{equation}
Q_{BL}(u; ς^-, ς^+, α, β) := Q_L\qty\big(Q_B(u)).
\end{equation}
The resulting quantile function is a slight generalization of the type IV generalized logistic distribution.

\subsubsection*{Risk neutralization}
Let $\wbreve{Q}(u) = Q_{BL}(u; ς^-, ς^+, α, β)$. The corresponding random variable $\wbreve{X} ∼ \wbreve{Q}^{-1}$ can be seen as a resampled logistic random variable whose quantile follows a beta distribution:
\begin{align}
\wbreve{X} = Q_L(V; ς^-, ς^+), &&
Q_B(U) =: V ∼ F_B(·; α, β), &&
U ∼ U[0,1].
\end{align}
Following Proposition~\ref{prop:neutr}, we calculate the moment generating function
\begin{align}
\E\qty[e^{λ \wbreve{X}}]
    &= \E\qty[e^{λ Q_L(V)}]
    = \E\qty[V^{λ ς^-} (1-V)^{-λ ς^+}] \\
    &= \frac{1}{B(α, β)}∫_0^1 v^{λς^- + α-1} (1-v)^{-λς^+ + β-1} \dd v 
    = \frac{B(α + λς^-, β - λς^+)}{B(α, β)},
\end{align}
which implies the integrability condition 
\begin{equation}
\E\qty[e^{λ \wbreve{X}}] < ∞ \text{ if and only if } β > ς^+.
\end{equation}
The risk neutral beta logistic quantile is given by
\begin{align}
Q(u) = μ + \wbreve{Q}(u), &&
μ = - \log \frac{B(α + ς^-, β - ς^+)}{B(α, β)}, &&
\wbreve{Q}(u) = Q_{BL}(u; ς^-, ς^+, α, β).
\end{align}

\subsubsection*{Pricing formula}
Similar to asymmetric logistic distribution, there is no closed form cumulative distribution function when $ς^- ≠ ς^+$. By setting $ς^± = ς$ we reduce the distribution to the one in \citep{carr2021additive, azzone2025explicit} with explicit
\begin{align}
F(x)
    &= Q_B^{-1} ∘ Q_L^{-1} ∘ \qty(\frac{x-μ}{ς})
    = F_B\qty\Big(F_L\qty\big(z(x; μ, ς)); α, β), \\
\wtilde{F}(x)
    &= ∫_0^{F(x)} e^{Q(u)} \dd u
    = ∫_0^{Q_B^{-1} ∘ Q_L^{-1} ∘ \qty(\frac{x-μ}{ς})} e^{μ + ς Q_L(Q_B(u))} \dd u
    = e^μ ∫_0^{F_L\qty\big(z(x; μ, ς))} e^{ς Q_L(v)} \dd F_B(v; α, β) \\
    &= F_B\qty\Big(F_L\qty\big(z(x; μ, ς)); α + ς, β - ς),
\end{align}
where $F_B$, $F^L$ and $z$ are as in equations~\eqref{eq:cdf_b} and \eqref{eq:cdf_l}. Plugging them into equation~\eqref{eq:price} recovers the pricing formula in \citep[Proposition 3.3]{azzone2025explicit}.

\subsubsection*{Wing asymptotic}
The tail asymptotics are
\begin{align}
T(x) ∼ 
\begin{cases}
    e^{-\frac{βx}{ς^+}}, & x → +∞ \\ 
    e^{\frac{αx}{ς^-}}, &  x → -∞,
\end{cases}
&&
\lim -\frac{\log T(x)}{\abs{x}} =
\begin{cases}
    \frac{β}{ς^+}, & x → +∞, \\
    \frac{α}{ς^-}, & x → -∞, \\
\end{cases}
\end{align}
Then by formula~\eqref{eq:iv_asym_lin} in Proposition~\ref{prop:asym}, the wing asymptotics are
\begin{equation}
\frac{ω^2(κ)}{\abs{κ}} ∼ 
\begin{cases}
ψ \qty(\frac{β}{ς^+} - 1), & κ → +∞, \\
ψ \qty(\frac{α}{ς^-}), & κ → -∞,
\end{cases}
\end{equation}
where the second line requires the additional condition $α > ς^-$.
\begin{remark}
Beta logistic distribution has linear wing asymptotics similar to those of logistic distribution, while the parameters $α$ and $β$ provide finer control over the shape of the distribution.
\end{remark}

Although the beta logistic model fits common U-shape curves well with moderate to long tenors, it lacks sufficient convexity to match the sharply rising wings of short-dated options, as indicated in the empirical results of \citep{lin2024neural}. Local concavity is also beyond its scope. Hence, a more flexible distribution is desired.

\subsection{Monotone transformation: Stretched exponential and logistic} 
Monotone transformation preserves the validity of a quantile function. Exponentiating the exponential quantile functions~\eqref{eq:qf_exp} with exponents $γ^+, γ^- > 0$ produces
\begin{align} \label{eq:qf_se}
Q^+(u; γ^+):= (-\log (1-u))^{\frac{1}{γ^+}}, &&
Q^-(u; γ^-):= - (-\log u)^{\frac{1}{γ^-}}, 
\end{align}
where $Q^+$ is the quantile function of stretched exponential distribution, also known as the Weibull distribution, with shape parameter $γ^+$, and $Q^-$ is the quantile function of negative Weibull distribution with shape parameter $γ^-$. We define $Q_{SL}$ by combining $Q^+$ and $Q^-$ and multiplying with the same scale $ς > 0$,
\begin{equation}
Q_{SL}(u; ς, γ^+, γ^-) = ς \qty\big(Q^+(u; γ^+) + Q^-(u; γ^-)) = ς \qty((-\log (1-u))^{\frac{1}{γ^+}} - (-\log u)^{\frac{1}{γ^-}}),
\end{equation}
$Q_{SL}$ can be regarded as a stretched logistic distribution, which contains logistic distribution when $γ^±=1$.

\subsubsection*{Risk neutralization} 
Let $\wbreve{Q}(u) = Q_{SL}(u; ς, γ^+, γ^-)$. The integrability of $\wbreve{X} ∼ \wbreve{Q}^{-1}$ follows from the known moment condition of the Weibull distribution:
\begin{equation} \label{eq:int_sl}
\E\qty[e^{λ \wbreve{X}}] < ∞ \text{ if and only if } γ^+, γ^- ≥ 1 \text{, where either equality requires additionally } ς < 1.
\end{equation}
In general, it has no closed form moment generating function, yet the shift $μ$ can be numerically calculated to obtain the risk neutralized version
\begin{align} \label{eq:qf_sl}
Q(u) = μ + \wbreve{Q}(u), &&
μ = - \log ∫_0^1 e^{Q_{SL}(u)} \dd u, &&
\wbreve{Q}(u) = ς Q_{SL}(u).
\end{align}

\subsubsection*{Pricing formula}
Stretched logistic distribution has no explicit cumulative distribution function, so $F$ and $\wtilde{F}$ have to be calculated through numerical inversion and numerical integration:
\begin{align} \label{eq:price_inv}
F(x) = Q^{-1}(x), &&
\wtilde{F}(x) = \int_0^{Q^{-1}(x)} e^{Q(u)} \dd u.
\end{align}

\subsubsection*{Wing asymptotic}
The stretched logistic has stretched exponential tails
\begin{align}
T(x) ∼
\begin{cases}
e^{-\qty(\frac{x}{ς})^{γ^+}}, & x → +∞ \\
e^{-\qty(\frac{-x}{ς})^{γ^-}}, & x → -∞.
\end{cases}
\end{align}
When $γ^± = 1$, it reduces to logistic distribution. When $γ^± > 1$, the wing asymptotics, inferred by formula~\eqref{eq:iv_asym_sub} in Proposition~\ref{prop:asym}, are sublinear:
\begin{equation} \label{eq:asym_sl}
\frac{ω^2(κ)}{\abs{κ}} ∼ 
\begin{cases}
\frac{ς^{γ^+}}{2}\abs{κ}^{1-γ^+}, & κ → +∞, \\
\frac{ς^{γ^-}}{2}\abs{κ}^{1-γ^-}, & κ → -∞.
\end{cases}
\end{equation}
\begin{remark}
This construction provides separate control over the left and right wing asymptotics of implied volatility with parameters $γ^+$ and $γ^-$ through a mechanism different from logistic or beta logistic distributions. $γ^± > 1$ provides full control over the sublinear wing asymptotics. For example, $γ^± = 2$ gives the flat implied volatility curve in extreme moneyness as in the Black-Scholes model and the Gaussian mixture model.
\end{remark}

The flexibility in sublinear asymptotics seems to act in the opposite direction from what is desired. Implied volatility curves produced by either logistic or beta logistic are already flat for short-dated options, so choosing $γ^± > 1$ makes the curve even flatter. Nevertheless, the asymptotics~\eqref{eq:asym_sl}, though not directly applicable, suggest that choosing $γ^± < 1$ may increase the convexity of the implied volatility curve, thereby motivating the following segmental construction.

\subsection{Piecewise construction: Quantile splice}
\begin{figure}[h]
    \centering
    \includegraphics[width=0.9\textwidth]{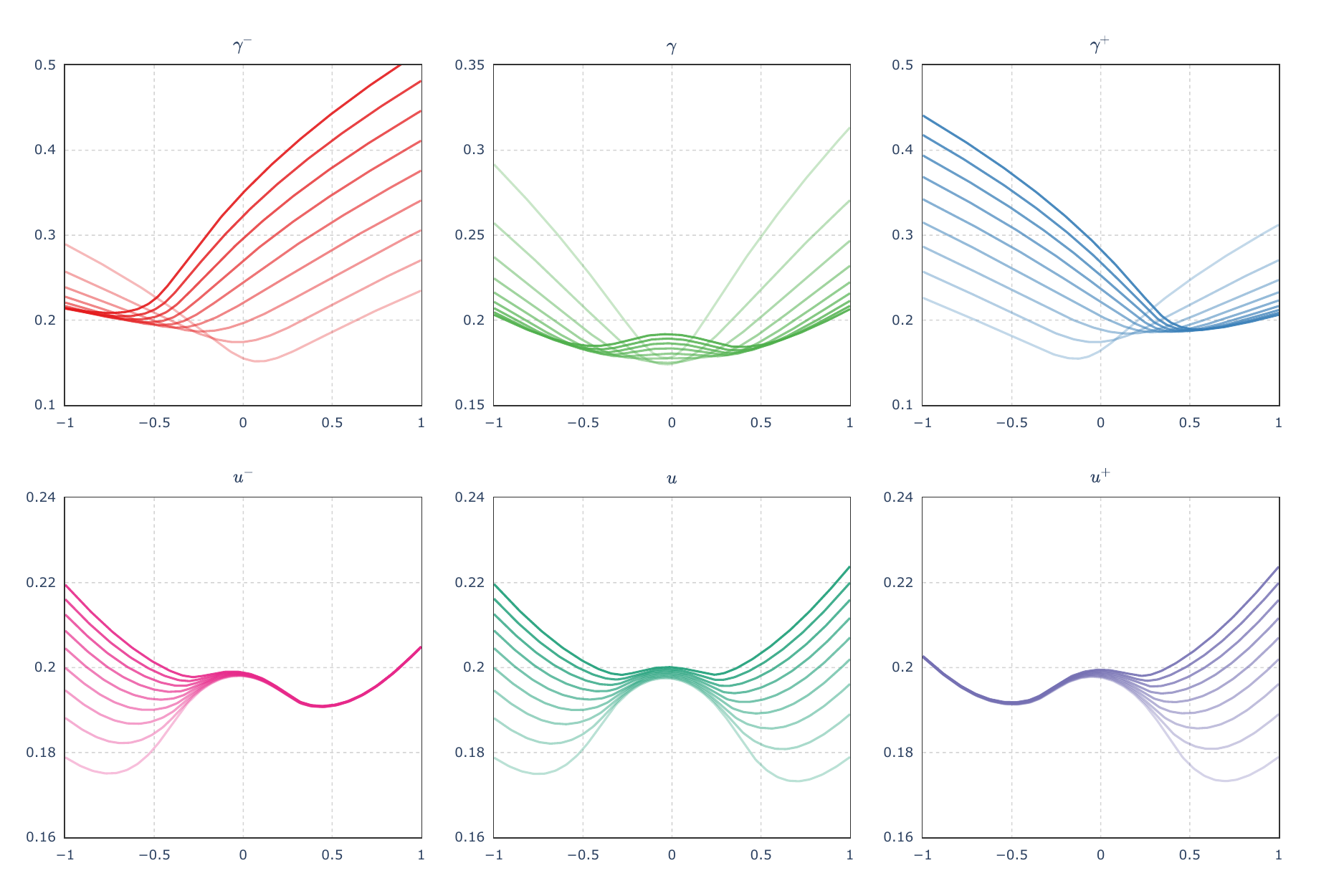}
    \caption{Implied volatility of risk neutral quantiles~\eqref{eq:qf_seqs2}. The base case uses parameters $\qty{ς, u^-, u^+, γ^-, γ^+} = \qty{0.1, 0.05, 0.95, 1, 1}$. The top row varies $γ ∈ \qty{0.8, 1.0, \dots, 2.4}$, with $ς ← 0.1γ$, as follows: (1) left: $γ^- ← γ$; (2) center: $γ^- = γ^+ ← γ$; and (3) right: $γ^+ ← γ$. The bottom row instead varies $u ∈ \qty{0.01, 0.02, \dots, 0.1}$, with $γ^- = γ^+ ← 3$, as follows: (4) left: $u^- ← u$; (5) center: $u^- = 1-u^+ ← u$; and (6) right: $1-u^+ ← u$. Lighter colors correspond to smaller values of the varied parameter.
    }
    \label{fig:demo_iv}
\end{figure}

Consider parameterizing two knots $u^-, u^+ ∈ (0,1)$ to split each of $Q^-$ and $Q^+$ into two segments as
\begin{align}
Q^-(u; u^-, γ^-) =
    \begin{cases}
    a^- + b^- \log u, & u ∈ [0, u^-), \\
    -\qty(-\log(u))^{\frac{1}{γ^-}}, & u ∈ [u^-, 1],
    \end{cases} &&
Q^+(u; u^+, γ^+) = 
    \begin{cases}
    \qty(-\log(1 - u))^{\frac{1}{γ^+}}, & u ∈ [0, u^+), \\
    a^+ + b^+ \qty(-\log (1-u)), & u ∈ [u^+, 1].
    \end{cases}
\end{align}
$a^±$ and $b^±$ are constants to guarantee the continuity and differentiability of $Q^±$ at $u^±$. We combine $Q^±$, scale their sum by $ς > 0$, and apply Proposition~\ref{prop:neutr} to create the risk neutral quantile function
\begin{equation} \label{eq:qf_seqs2}
    Q(u; ς, u^-, u^+, γ^-, γ^+) = μ + ς \qty(Q^-(u; u^-, γ^-) + Q^+(u; u^+, γ^+)).
\end{equation}
Figure~\ref{fig:qf} (center) illustrates the two-segment construction. $Q$ has exponential tails~\eqref{eq:qf_exp} in the tail segments $[0, u^-)$ and $[u^+, 1]$, and behaves as the stretched logistic quantile~\eqref{eq:qf_sl} in the remaining segment $[u^-, 1]$ and $[0, u^+)$ yet with fewer restrictions on the choice of $γ^±$. More details such as the integrability condition will be discussed in Section~\ref{sec:seqs}.

Figure~\ref{fig:demo_iv} demonstrates how the five intuitive parameters - $ς$ controls global curve level, $u^±$ determine segment intervals, and $γ^±$ affect curve slope and convexity - produce a broad spectrum of implied volatility curves, covering U-shape, inverse U-shape, W-shape, and S-shape. For the symmetric case shown in the second column, the number of parameters is reduced to three by $u^-=1-u^+$, $γ^-=γ^+$.

\begin{remark}
Compared to the Gaussian mixture, which requires three components to produce a W-shape \citep{glasserman2023w}, or the variance gamma mixture, which requires two \citep{keller2025w}, the model~\eqref{eq:qf_seqs2} requires only three parameters, demonstrating pronounced expressivity.
\end{remark}

We have hitherto crafted a family of flexible quantile functions, which we call \emph{stretched exponential quantile splice} for two reasons. First, the entire quantile is piecewise and smoothly constructed, in a manner similar to, but distinct from, that of a spline. Second, the basis in each segment is a quantile function, specifically, the quantile function of stretched exponential distribution. The stretched exponential quantile splice covers the logistic quantile~\eqref{eq:qf_l} and its stretched generalization~\eqref{eq:qf_sl} as single-segment instances, and the quantile function~\eqref{eq:qf_seqs2} as a two-segment instance. Intuitively, it can produce more complicated shapes with more segments. We henceforth formalize the model and further investigate it.

\section{Stretched exponential quantile splice} \label{sec:seqs}
\subsection{Formulation}
Let $M ∈ \N$ be the number of segments, and we sometimes refer to the $M$-segment stretched exponential quantile splice as the \emph{$M$-segment model} for short. It is formulated as
\begin{equation} \label{eq:qf_seqs}
\begin{cases}
Q(u)
    = μ + ς \qty(Q^-(u) + Q^+(u)),\\
Q^±(u)
    = \sum_{m=0}^{M-1} \one_{I_m^±} Q^±_m(u), \\
Q^±_m
    = a_m^± + b_m^± \qty(l^±(u))^{\frac{1}{γ_m^±}},   \text{for }0 ≤ m ≤ M -1,
\end{cases}
\end{equation}
and its quantile density function is
\begin{equation} \label{eq:qdf_seqs}
\begin{cases}
q(u) 
    = ς \qty(q^-(u) + q^+(u)),\\
q^±(u)
    = \sum_{m=0}^{M-1} \one_{I_m^±} q^±_m(u), \\
q^±_m
    = ± \frac{b_m^±}{γ^±_m} e^{l^±(u)}\qty(l^±(u))^{\frac{1}{γ_m^±} - 1},   \text{for }0 ≤ m ≤ M - 1.
\end{cases}
\end{equation}
It has $4M - 1$ parameters - a scale $ς$, $2(M-1)$ interior points $u^±_m$, and $2M$ shape parameters $γ^±_m$:
\begin{align} \label{eq:param_seqs}
θ = \qty{ς} ∪ \qty{u^±_m}_{m=1}^{M-1} ∪ \qty{γ^±_m}_{m=0}^{M-1},
&& &ς > 0,   u^±_m ∈ [0,1],   γ^±_m > 0.
\end{align}
We will write the dependence on parameters explicitly, such as $Q(u; θ)$, when necessary. One can conceptualize the model through the $M=3$ instance as illustrated in Figure~\ref{fig:qf} (right) for intuition.

\begin{remark}
One may also choose different numbers of segments $M^±$ and different scales $ς^±$ for $Q^±$.
\end{remark}

\begin{remark} \label{rema:param_seqs}
In practice, the tail segment is less relevant as options of extreme moneyness are illiquid, so for $M≥2$, $u^±_1$ can be fixed near the boundary, and $γ^±_0$ can be set to a proper constant, which presets the wing asymptotics and reduces the number of parameters from $4M - 1$ to $4M - 5$.
\end{remark}

In equations~\eqref{eq:qf_seqs} and \eqref{eq:qdf_seqs}, we use the following shorthand notation for logarithms
\begin{align}
l^+(u) := - \log(1-u),
&&
l^-(u) := - \log u,
&&
l^±_m := l^±(u^±_m).
\end{align}
The interior points $u^±_m$ are ordered as follows to create disjoint intervals $I^±_m$:
\begin{align}
0 ≤ u^-_1 ≤ … ≤ u^-_{M-1} ≤ 1,
&& 
[0,1] = \bigcup_{m=0}^{M-1} I^-_m = [0, u^-_1) ∪ [u^-_1, u^-_2)∪ … ∪ [u^-_{M-2}, u^-_{M-1}) ∪ [u^-_{M-1}, 1], \\
0 ≤ u^+_{M-1}  ≤ …  ≤ u^+_1 ≤ 1, 
&&
[0,1] = \bigcup_{m=0}^{M-1} I^+_m = [0, u^+_{M-1}) ∪ [u^+_{M-1}, u^+_{M-2})∪ … ∪ [u^+_2, u^+_1) ∪ [u^+_1, 1].
\end{align}
Each segment $Q^±_m$ in $I^±_m$ is a shifted Weibull quantile function with location constant $a^±_m$, scale constant $b^±_m$, and shape parameter $γ^±_m$. Constants $a^±_m$ and $b^±_m$ are set to enforce smoothness and monotonicity of $Q^±$: 
\begin{align} \label{eq:smooth_seqs}
Q^±_m(u_{m+1}^±) = Q_{m+1}^±(u_{m+1}^±),
&&
q^±_m(u_{m+1}^±) = q_{m+1}^±(u_{m+1}^±),
&&
\text{for } 0 ≤ m ≤ M-2.
\end{align}
Equation~\eqref{eq:smooth_seqs} is underdetermined as each of $Q^±$ has one more segment than the number of knots. We choose to fix $a^±_{M-1}=0$ and $b^±_{M-1}=1$ to anchor at the finite segment $Q^±_{M-1}$. Substituting equations~\eqref{eq:qf_seqs}-\eqref{eq:param_seqs} into \eqref{eq:smooth_seqs} then yields
\begin{equation} \label{eq:const_seqs}
\begin{aligned}
b^±_m
&=
    \begin{cases}
        1, & m = M-1, \\
        \frac{γ^±_m b^±_{m+1}}{γ^±_{m+1}} \qty(l^±_{m+1})^{\frac{1}{γ^±_{m+1}} - \frac{1}{γ^±_m}}, & 0 ≤ m ≤ M-2,
    \end{cases}
\\
a^±_m
&=
    \begin{cases}
        0, & m = M-1, \\
        a^±_{m+1} + b^±_{m+1}\qty(l^±_{m+1})^{\frac{1}{γ^±_{m+1}}} - b^±_m\qty(l^±_{m+1})^{\frac{1}{γ^±_m}}, & 0 ≤ m ≤ M-2.
    \end{cases}    
\end{aligned}
\end{equation}
We explicitly express the constant in the tail segment
\begin{equation}
b^±_0 = \prod_{m=0}^{M-2} \frac{γ^±_m}{γ^±_{m+1}} \qty(l^±_{m+1})^{\frac{1}{γ^±_{m+1}} - \frac{1}{γ^±_m}},
\end{equation}
which will be used in Proposition~\ref{prop:int_seqs} and \ref{prop:asym_seqs} for risk neutrality and wing asymptotics.

\begin{remark}
It is optional to anchor at the tail segments $Q^±_0$ with $a^±_0 = 0$ and $b^±_0 = 1$, which simplifies the integrability condition~\eqref{eq:rint_seqs} below. However, anchoring at the main segments $Q^±_{M-1}$ offers more practical control over the shape of implied volatility even without an optimizer, which will be shown soon in Section~\ref{sec:init}.
\end{remark}

\subsection{Risk neutralization}
\begin{proposition}[Risk neutralization] \label{prop:int_seqs}
Let $Q$ be a stretched exponential quantile splice as in equation~\eqref{eq:qf_seqs} with parameters~\eqref{eq:param_seqs} and constants~\eqref{eq:const_seqs}, and let $X ∼ Q^{-1}$. Then
\begin{enumerate}
    \item $\E\qty[e^X] < ∞$ if and only if
        \begin{equation} \label{eq:rint_seqs}
        γ^+_0 > 1,   \text{or}   γ^+_0 = 1 \text{ with } ς b^+_0 < 1.
        \end{equation}
    \item $\E\qty[e^X] = 1$ if and only if condition~\eqref{eq:rint_seqs} holds and
        \begin{equation} \label{eq:mu_seqs}
        μ = - \log ∫_0^1 e^{ς\qty(Q^-(u) + Q^+(u))} \dd u.
        \end{equation}
\end{enumerate}
\end{proposition}

\begin{proof}
Since $\E\qty[e^X] = ∫_0^1 e^{Q(u)} \dd u = e^μ ∫_0^1 e^{ς \qty(Q^-(u) + Q^+(u))} \dd u$ and $Q^-(u)$ is bounded near $u=1$, the integrability is determined by the first segment $ςQ^+_0(u) = ςa^+_0 + ςb^+_0 \qty(l^+(u))^{\frac{1}{γ^+_0}}$, a Weibull quantile with location $ςa^+_0$, scale $ςb^+_0$ and shape $γ^+_0$. Thus the integrability condition~\eqref{eq:rint_seqs} follows directly from the classical moment property of Weibull distribution, and the risk neutrality condition~\eqref{eq:mu_seqs} follows from Proposition~\ref{prop:neutr}.
\end{proof}

Hereafter, when we refer to \emph{stretched exponential quantile splice} or simply the \emph{model}, we specifically mean the risk neutral one that satisfies conditions~\eqref{eq:rint_seqs} and \eqref{eq:mu_seqs}. Similar to the stretched logistic distribution~\eqref{eq:qf_sl}, the model has no explicit cumulative distribution function, so option pricing requires numerical inversion and numerical integration as in equation~\eqref{eq:price_inv}.

\subsection{Wing asymptotic}
Wing asymptotics are controlled by parameter $γ^±_0$ and constant $b^±_0$.
\begin{proposition}[Wing asymptotic] \label{prop:asym_seqs}
Let $ω$ denote the implied volatility generated by a risk neutral stretched exponential quantile splice, then
\begin{enumerate}
    \item
    For the right wing, as $κ → +∞$,
        \begin{enumerate}
            \item if $γ^+_0 > 1$, then
            \begin{equation}
                \frac{ω^2(κ)}{κ} ∼ 
                \frac{ς b^+_0}{2} κ^{1 - γ^+_0};
            \end{equation}
            \item if $γ^+_0 = 1$ and $ς b^+_0 < 1$, then
            \begin{equation}
                \frac{ω^2(κ)}{κ} ∼ ψ\qty(\frac{1}{ς b^+_0} - 1).
            \end{equation}
        \end{enumerate}
    \item For the left wing, as $κ → -∞$,
        \begin{enumerate}
        \item if $γ^-_0 > 1$, then
            \begin{equation}
                \frac{ω^2(κ)}{\abs{κ}} ∼ \frac{ς b^-_0}{2} \abs{κ}^{1 - γ^-_0};
            \end{equation}
        \item if $γ^-_0 = 1$ and $ς b^-_0 < 1$, then
            \begin{equation}
                \frac{ω^2(κ)}{\abs{κ}} ∼ ψ\qty(\frac{1}{ς b^-_0});
            \end{equation}
        \item otherwise,
            \begin{equation}
                \limsup \frac{ω^2(κ)}{\abs{κ}} = 2.
            \end{equation}
        \end{enumerate}
\end{enumerate}
\end{proposition}
\begin{proof}
It follows immediately from Proposition~\ref{prop:asym}.
\end{proof}

\section{Empirical study} \label{sec:empirical}
With the stretched exponential quantile splice model fully specified, we now apply it to the market dataset. We demonstrate that a seven-parameter, three-segment specification is sufficient to calibrate the vast majority of implied volatility curves in the dataset. The small number of curves exhibiting uncommon shapes can be accommodated either by introducing additional segments or by applying certain perturbations.

\subsubsection*{Dataset}
The Bloomberg dataset contains 4,969,256 pairs of SPX and SPXW put and call options quoted over the 502 trading days during 2024 and 2025, covering 25,770 implied volatility curves that constitute 1,004 implied volatility surfaces. For time notation, we use $t$ to denote the as-of trading day, and $τ$ for \emph{tenor} defined as the number of days to expiration divided by $365.25$. We denote the total tenor set by $\cT := [0, T]$ for some $T$, the total moneyness set by $\cK := \R$, and the quantile interval by $\cU := (0,1)$. Let
$$
\cD = \qty{ \text{trading days from 2024-01-01 to 2025-12-31}},
$$
At $t ∈ \cD$, we denote its set of tenors by $\cT_t$, and the set of listed moneynesses at $t$ with tenor $τ$ by $\cK_{t, τ}$.

For quantities quoted on day $t$ with tenor $τ$ and moneyness $κ$, we denote the mid implied volatility by $σ_t(τ, κ)$, which corresponds to the mid price of an out of the money option $o_t(τ, κ) = p_t(τ, κ) \one_{κ ≤ 0} + c_t(τ, κ) \one_{κ > 0}$. We use $\overline{σ}_t(τ, κ)$ and $\underline{σ}_t(τ, κ)$ to denote the ask and bid implied volatility.

\subsubsection*{Three-segment model}
We start with a three-segment risk neutral stretched exponential quantile splice, which is model~\eqref{eq:qf_seqs} with $M=3$, to fit all the implied volatility curves. As noted in Remark~\ref{rema:param_seqs}, we fix the tail segment and wing asymptotics by
\begin{align}
u^-_1 = 10^{-12}, &&
u^+_1 = 1 - 10^{-12}, &&
γ^±_0 = 1.
\end{align}
Hence, there remain, parsimoniously, seven free parameters:
\begin{align} \label{eq:param_seqs3}
θ=\qty{ς, u^-_2, u^+_2, γ^-_1, γ^-_2, γ^+_1, γ^+_2}.
\end{align}
\begin{remark}
This parametrization conveys the following intuition. The extreme moneyness quantile intervals $[0, u^-_1)$ and $[u^+_1, 1]$ are practically irrelevant and thus we choose arbitrarily $γ^+_0$ to impose linear asymptotic. The deep out of the money intervals $[u^-_1, u^-_2)$ and $[u^+_2, u^+_1)$, as we will soon see from the empirical evidence, are thin yet influential - each has a size of about 0.02 to 0.05 - and they control the visible wing. The values of $γ^±_1$ in these intervals are small, corresponding to steeper wing slope. The remaining intervals $[u^-_2, 1]$ and $[0, u^+_2]$ are wider, where $γ^±_2$ has a larger value corresponding to a flatter wing slope.
\end{remark}

Let $o(κ; θ)$ denote the price of out of the money option with moneyness $κ$ given by the three-segment model with parameter $θ$, $ω(κ; θ)$ denote its implied total volatility, and $σ(τ, κ; θ) = ω(κ; θ) / \sqrt{τ}$ denote the model implied volatility.

\subsubsection*{Calibration objective}
For every day $t ∈ \cD$ and tenor $τ ∈ \cT_t$, we fit each implied volatility curve independently. The calibration minimizes the weighted root mean square difference between the model price and the mid market price:
\begin{align} \label{eq:optim}
\what{θ}_t(τ) = \argmin_θ \sqrt{\frac{1}{\abs{\cK_{t, τ}}} \sum_{κ ∈ \cK_{t, τ}} \frac{1}{w_t(τ, κ)}\qty(o(κ; θ) - o_t(τ, κ))^2}, &&
\text{subject to } ςb^+_0 < 1,
\end{align}
Possible choices of weights $w_t(τ, κ)$ include squared vega $ν_t^2(τ,κ)$, vega $ν_t(τ,κ)$, implied volatility spread $\overline{σ}_t(τ, κ) - \underline{σ}_t(τ, κ)$, volume, or other variations; see \citep{homescu2011implied}. The vega is calculated as $ν_t(τ,κ) := ∂_σ o_t(τ, κ) = ϕ(z^-(κ, σ_t(τ, κ))) \sqrt{τ}$.

In our experiments, both weighting methods yield satisfactory fits when data quality is high and the curve lies within the model’s capacity. However, the calibration results begin to diverge across weighting schemes when the data are noisy or the curve shape is more challenging for the model. By default, we report the results using the squared vega weight $ν_t^2(τ,κ)$, and will consider alternative weights for specific samples in Section~\ref{sec:refine}.

Optimization can be performed using standard algorithms. Before doing so, however, we revisit the motivation outlined in the abstract - namely, why these seven parsimonious parameters are descriptive and how the model affords modular, practical control over implied volatility - by handcrafting a calibration without numerical optimizers.

\subsection{Handcrafting a fit} \label{sec:init}
We pick the first trading day in the dataset, $t=\text{2024-01-02}$. It contains $54$ expiration dates from 2024-01-03 to 2025-06-20, with corresponding tenors ranging from $0.0029$ to $1.4642$. We select the middle expiration date 2024-02-09, about one month and one week until expiration, tenor $τ=0.1042$. The implied volatility presents the typical U-shape as shown in Figure~\ref{fig:demo_fit}. The manual calibration proceeds progressively as follows:
\begin{figure}[h]
    \centering
    \includegraphics[width=1.0\textwidth]{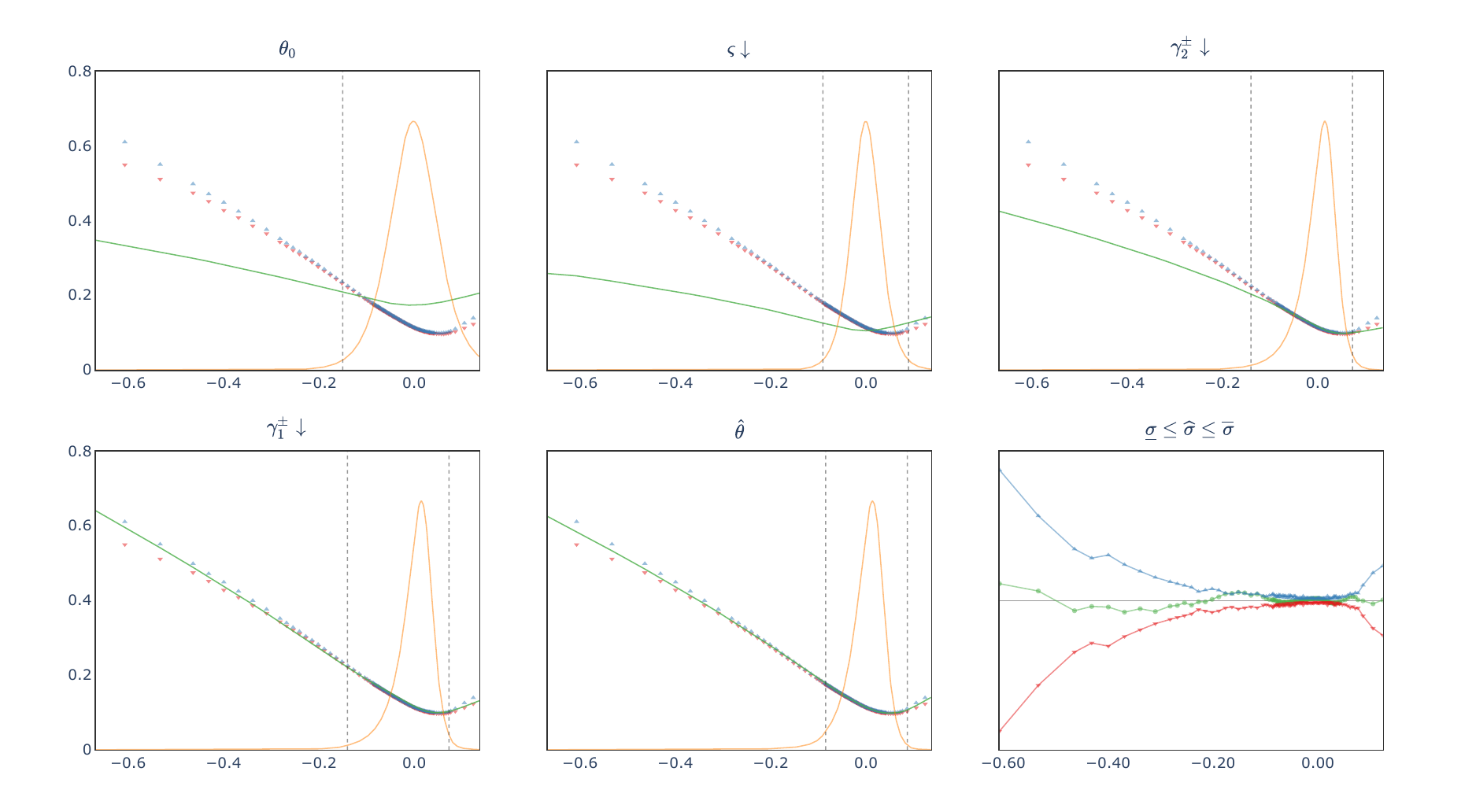}
    \caption{Calibrating a three-segment model to market data on 2024-01-02 with expiration 2024-02-09. Horizontal axis: moneyness; vertical axis: value. In the first five panels, red and blue scatter points are bid and ask implied volatility; green and orange curves are model implied volatility and implied density; left and right vertical dotted lines show the $u^-_2$ and $u^+_2$ quantiles of moneyness. In the last panel, red, green, and blue curves represent the bid-to-mid, model-to-mid, and ask-to-mid implied volatility differences.}
    \label{fig:demo_fit}
\end{figure}

\begin{enumerate}
    \item Initialize from a simple logistic distribution with parameter
    \begin{align}
    θ 
        = \qty{ς, u^-_2, u^+_2, γ^-_1, γ^-_2, γ^+_1, γ^+_2}
        = \qty{0.10\sqrt{τ}, 0.01, 0.99, 1, 1, 1, 1}.
    \end{align}
    Figure~\ref{fig:demo_fit} (top left) shows the market bid and ask implied volatility across moneynesses, and the model implied volatility and implied density.
    \item Lower the overall implied volatility level by reducing the scale $ς$. Figure~\ref{fig:demo_fit} (top center) is produced by $ς ← 0.06\sqrt{τ}$.
    \item Adjust the convexity near the money. Since the model implied volatility skew on the left is weaker than the market data, we raise it by decreasing $γ^-_2$ while keeping $γ^-_1 ≤ γ^-_2$. In contrast, we relax the right skew that is higher than the market data by increasing $γ^+_2$. This leads to $γ^-_1 = γ^-_2 ← 0.76$ and $γ^+_2 ← 1.2$, and the model implied volatility in Figure~\ref{fig:demo_fit} (top right) matches the market data near the money.
    \item Raise the out of the money skew on both left and right sides by further decreasing $γ^-_1$ and $γ^+_1$. With $γ^-_1 ← 0.25$ and $γ^+_1 ← 0.5$, Figure~\ref{fig:demo_fit} (lower left) shows that the model implied volatility matches the market data. At this stage, the model parameter vector is
    \begin{align} \label{eq:param_init}
    θ ← \qty{0.06 \sqrt{τ}, 0.01, 0.99, 0.25, 0.76, 0.5, 1.2}.
    \end{align}
    \item Adjust segmentation with $u^±_2$. Repeating these steps to refine the parameters, we eventually reach
    \begin{align}
    \what{θ}_t(τ) = \qty{0.06 \sqrt{τ}, 0.03, 0.99, 0.33, 0.81, 0.24, 1.22}, &&
    t = \text{2024-01-02}, &&
    τ = \frac{\text{2024-02-09} - \text{2024-01-02}}{365.25}.
    \end{align}
    Figure~\ref{fig:demo_fit} (lower center) shows the optimal model implied volatility and implied density. Figure~\ref{fig:demo_fit} (lower right) confirms that the deviation between the model and the mid implied volatility lies well within the bid-ask spread and oscillates around zero. Thus $\what{θ}_t(τ)$ provides an excellent fit.
\end{enumerate}

\subsection{Parameter term structure} \label{sec:param_seqsterm}
\begin{figure}
    \centering
    \includegraphics[width=1.0\linewidth]{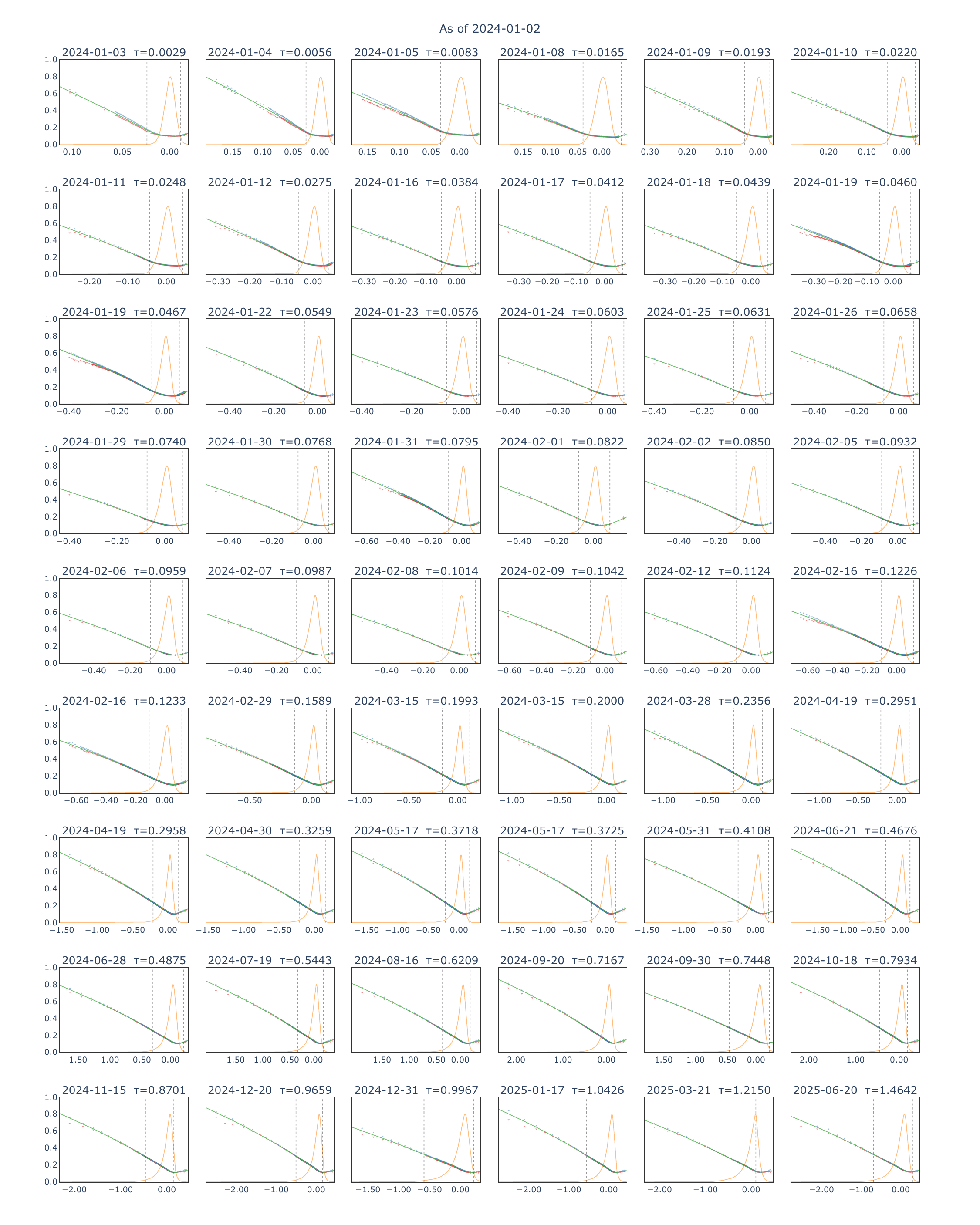}
    \caption{Fitted three-segment model~\eqref{eq:qf_seqs} for 2024-01-02 over all listed tenors.}
    \label{fig:iv_20240102}
\end{figure}

The parameter guess~\eqref{eq:param_init} for a representative U-shape curve is a good initial parameter to start the optimization~\eqref{eq:optim}. We therefore calibrate all implied volatility curves separately using this initialization. To complete the previous example in Section~\ref{sec:init}, we report the calibration results for 2024-01-02 across all 54 expirations in Figure~\ref{fig:iv_20240102} to show the goodness of fit.

\begin{figure}
    \centering
    \includegraphics[width=1.0\textwidth]{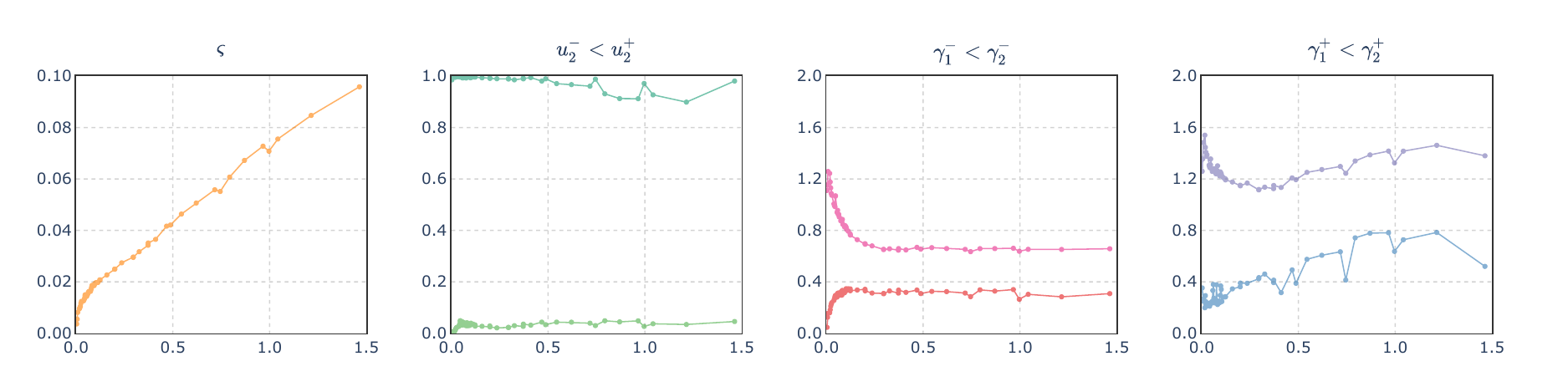}
    \caption{Fitted parameters~\eqref{eq:param_seqs3t} for 2024-01-02 over all listed tenors. Horizontal axis is tenor.}
    \label{fig:demo_param}
\end{figure}

Beyond precise individual curve calibration, enabling interpolation of curves into a whole implied volatility surface is crucial for further downstream tasks like constructing dynamic models. This requires parameter stability. For a given trading day, consider the parameter term structure as a function of tenor,
\begin{align} \label{eq:param_seqs3t}
    θ(τ) = \qty{ς(τ), u^-_2(τ), u^+_2(τ), γ^-_1(τ), γ^-_2(τ), γ^+_1(τ), γ^+_2(τ)}: \cT → \R_+ × \cU^2\times \R_+^4.
\end{align}
Then for day $t$, the calibrated parameters as discrete observations of the term structure across tenors, $\qty{\what{θ}_t(τ)}_{τ ∈ \cT_t}$, are expected to form a smooth pattern.

Figure~\ref{fig:demo_param} presents a clear pattern of the fitted parameters of day 2024-01-02 across tenors. Scale $ς(τ)$ increases with tenor, analogous to the accumulation of total volatility. Segmentation knots $u^-_2(τ)$ and $u^+_2(τ)$ lie stably near the levels $0.03$ and $0.98$. Shape parameters for deep out of the money intervals, $γ^±_1$, have smaller values for short tenors, corresponding to sharply rising wings; their increase with tenor indicates the wing becomes flatter for longer tenors. The shape parameter for the major intervals, $γ^±_2$, on the other hand, is large at short tenors and evolves differently as tenor increases, indicating curve asymmetry.

Appendix~\ref{sec:appendix} collects extra calibrations with more curve shapes, including curves with local concavity.

\subsection{Model evaluation}
\begin{figure}
    \centering
    \includegraphics[width=0.8\linewidth]{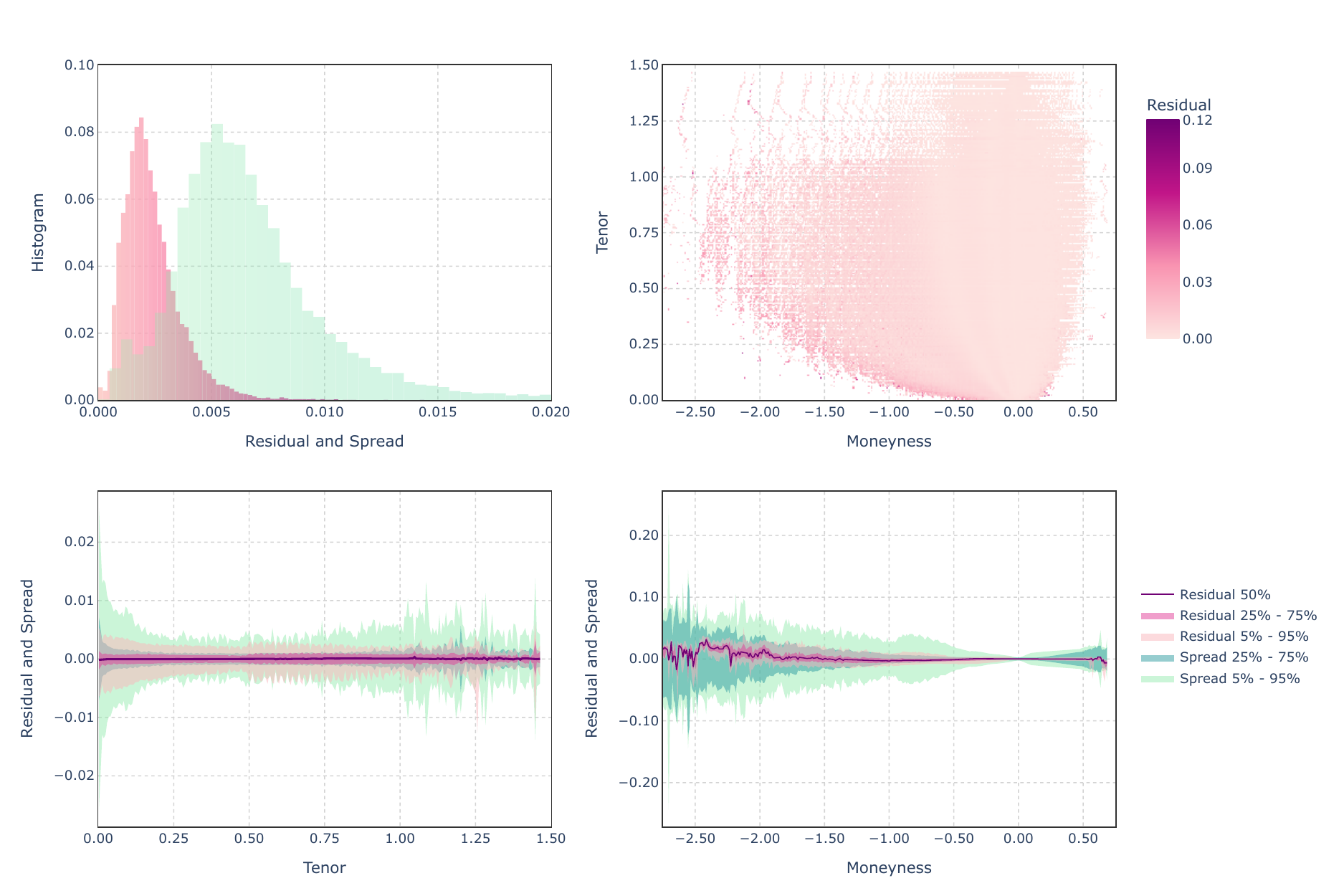}
    \caption{Root mean square residuals (red) versus root mean square spread (green). Top left: curve-wise~\eqref{eq:resid_curve}. Top right: quotation-wise~\eqref{eq:resid_quote}. Lower left: tenor-wise~\eqref{eq:resid_tenor}. Lower right: moneyness-wise~\eqref{eq:resid_money}.}
    \label{fig:loss}
\end{figure}

The overall calibration performance of the 25,770 implied volatility curves is evaluated directly through root mean square residuals between the model implied volatility and the mid implied volatility, instead of weighted option price difference used in calibration. The root mean square (half) spread will be used as a reference.

\subsubsection*{Curve-wise performance}
For every $t ∈ \cD$ and $τ ∈ \cT_t$, we calculate the curve-wise root mean square residual and root mean square spread
\begin{equation} \label{eq:resid_curve}
\begin{aligned}
ε_t(τ) &:= \sqrt{\frac{1}{\abs{\cK_{t,τ}}} \sum_{κ ∈ \cK_{t,τ}} \qty(σ(κ; \what{θ}_t(τ)) - σ_t(τ, κ))^2}, \\
δ_t(τ) &:= \sqrt{\frac{1}{\abs{\cK_{t,τ}}} \sum_{κ ∈ \cK_{t,τ}} \frac{1}{4}  \qty(\overline{σ}_t(τ, κ) - \underline{σ}_t(τ, κ))^2}.
\end{aligned}
\end{equation}
Figure~\ref{fig:loss} (top left) shows a histogram of residuals $\qty{ε_t(τ)}_{t\in \cD, τ ∈ \cT_t}$ in red and spreads $\qty{δ_t(τ)}_{t\in \cD, τ ∈ \cT_t}$ in green. The former is concentrated mainly in a narrower interval to the left of the latter, indicating that the fitted curves mostly lie within the spreads.

\subsubsection*{Quotation-wise performance}
To assess whether the model exhibits systematic calibration bias, we further examine the residuals grouped by quotation, tenor, and moneyness. The dataset spans tenors $τ ∈ [0.0022, 1.4642]$ and moneynesses $κ ∈ [-2.87, 0.69]$. For grouping, we partition the tenor set into $n_{\cT} = 270$ bins with uniform step $Δ_τ = 0.005$, $\bigcup_{i=0}^{n_{\cT}-1} I^{\cT}_i = \bigcup_{i=0}^{n_{\cT}-1} [0.002 + iΔ_τ, 0.002 + (i+1) Δ_τ)$, and the moneyness set into $n_{\cK}=358$ bins with $Δ_κ = 0.01$, $\bigcup_{j=0}^{n_{\cK}-1} I^{\cK}_j = \bigcup_{j=0}^{n_{\cK}-1} [-2.87 + jΔ_κ, -2.87 + (j+1) Δ_κ)$. We then calculate the $n_{\cT} × n_{\cK} = 96{,}660$ quotation-wise root mean square residuals for $i ∈ \qty{0, 1, \dots, n_{\cT}-1}$ and $j ∈ \qty{0, 1, \dots, n_{\cK}-1}$:
\begin{equation}  \label{eq:resid_quote}
\begin{aligned}
ε_{\cT,\cK} (i,j)
    &:= \sqrt{\frac{1}{n\qty(I^{\cT}_i,I^{\cK}_j)} \sum_{t ∈ \cD, τ ∈ I^{\cT}_i, κ ∈ I^{\cK}_j}\qty(σ(κ; \what{θ}_t(τ)) - σ_t(τ, κ))^2},
\end{aligned}
\end{equation}
where $n\qty(I^{\cT}_i,I^{\cK}_j)$ counts the number of options in quotation group $I^{\cT}_i × I^{\cK}_j$.

The heatmap in Figure~\ref{fig:loss} (top right) shows that the quotation-wise residuals $ε_{\cT,\cK} (i,j)$ are slightly higher for large moneyness and short tenor, an expected pattern in implied volatility calibration results. Otherwise, the residual level is generally small.

\subsubsection*{Tenor-wise performance}
We calculate the tenor-wise root mean square residual by aggregating over the $n_{\cT}$ tenor bins, for $i ∈ \qty{0, 1, \dots, n_{\cT}-1}$:
\begin{equation} \label{eq:resid_tenor}
\begin{aligned}
ε_{\cT} (i)
    &:= \sqrt{\frac{1}{n(I^{\cT}_i)} \sum_{t ∈ \cD, τ ∈ I^{\cT}_i, κ ∈ \cK_{t,τ}}\qty(σ(κ; \what{θ}_t(τ)) - σ_t(τ, κ))^2}, \\
δ_{\cT} (i)
    &:= \sqrt{\frac{1}{n(I^{\cT}_i)} \sum_{t ∈ \cD, τ ∈ I^{\cT}_i, κ ∈ \cK_{t,τ}} \frac{1}{4}  \qty(\overline{σ}_t(τ, κ) - \underline{σ}_t(τ, κ))^2},
\end{aligned}
\end{equation}
where $n\qty(I^{\cT}_i)$ counts the number of options in tenor group $I^{\cT}_i$.

Figure~\ref{fig:loss} (lower left) visualizes the median tenor-wise residual, together with the $(25\%, 75\%)$ and $(5\%, 95\%)$ bands for the residual and the spread. Both bands are wider at short tenors, while the residual bands consistently lie within the spread bands. This result indicates that larger quotation-wise residuals at short tenors in Figure~\ref{fig:loss} (top right) are not driven by systematic deviation.

\subsubsection*{Moneyness-wise performance}
The moneyness-wise root mean square residual, similar to the tenor-wise one, is then aggregated across the $n_{\cK}$ bins, $j ∈ \qty{0, 1, \dots, n_{\cK}-1}$:
\begin{equation} \label{eq:resid_money}    
\begin{aligned} 
ε_{\cK} (j) 
    &:= \sqrt{\frac{1}{n(I^{\cK}_j)} \sum_{t ∈ \cD, τ ∈ \cT_t, κ ∈ I^{\cK}_j}\qty(σ(κ; \what{θ}_t(τ)) - σ_t(τ, κ))^2}, \\
δ_{\cK} (j)
    &:= \sqrt{\frac{1}{n(I^{\cK}_j)} \sum_{t ∈ \cD, τ ∈ \cT_t, κ ∈ I^{\cK}_j} \frac{1}{4}  \qty(\overline{σ}_t(τ, κ) - \underline{σ}_t(τ, κ))^2},
\end{aligned}
\end{equation}
where $n\qty(I^{\cK}_j)$ counts the number of options in moneyness group $I^{\cK}_j$.

Figure~\ref{fig:loss} (lower right) plots the median moneyness-wise residual, together with the $(25\%, 75\%)$ and $(5\%, 95\%)$ bands of the residual and the spread. While both bands widen at larger values of moneyness, the residual bands are contained within the spread bands. This shows that the higher quotation-wise residual in Figure~\ref{fig:loss} (top right) is not due to systematic misfit.

Conceptually, the above metrics serve as a multiview projection of the average residuals and spreads, with the residuals lying well within the spreads without systematic deviation, indicating that the three-segment model produces precise calibrations to market data. 

\subsection{Outlier analysis and refinement} \label{sec:refine}
To outline the capability boundary of the three-segment model, we closely examine implied volatility curves that contribute to the top few hundred root mean square residuals - those in the right tail of residual histogram in Figure~\ref{fig:loss} (top left). 

A substantial portion of the large residuals arises from market microstructure noise in the form of wide bid-ask spreads and quote variability, where fits that lie well within the bid-ask spread may still exhibit large residual values. Excluding such samples from the analysis, we observe that the genuine misfit occurs in the presence of unusually large concavity in the implied volatility, thereby outlining the capability limits of the three-segment model.

\begin{figure}[h]
    \centering
    \includegraphics[width=1.0\textwidth]{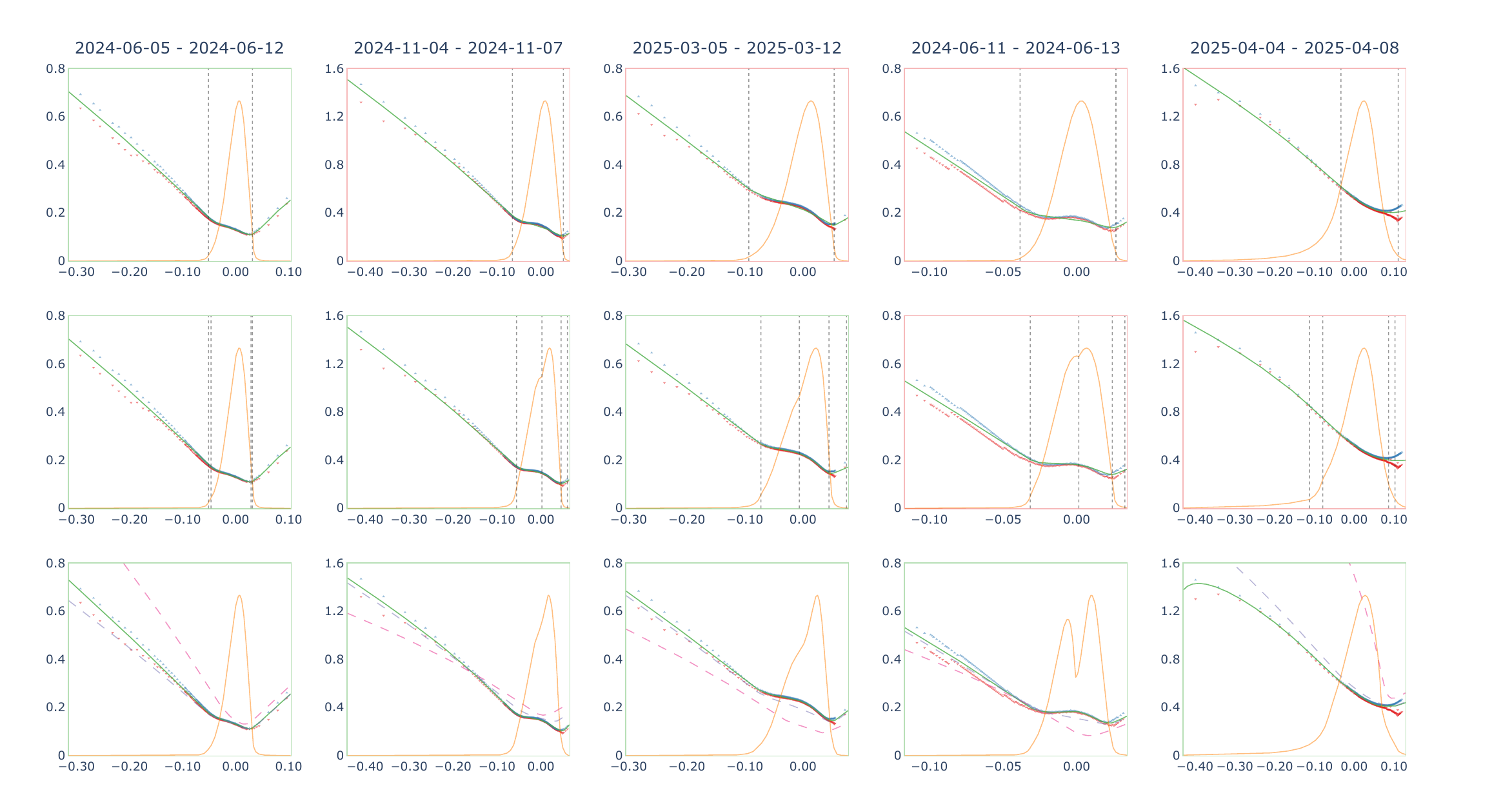}
    \caption{Refining worst fits. Each column presents fits by three models to the same data. First row: three-segment model. Second row: four-segment model. Third row: hybrid model. A green frame indicates good fit and a red frame indicates misfit. The graphical conventions are the same as in Figure~\ref{fig:demo_fit}. The second row shows additional vertical lines corresponding to the $u^-_2 < u^-_3 < u^+_3 < u^+_2$ quantiles. The third row instead shows two component implied volatilities as dashed purple and magenta curves.}
    \label{fig:refine}
\end{figure}

\subsubsection*{Worst fits of the three-segment model}
Figure~\ref{fig:refine} collects five illustrative concave implied volatility curves. They are ordered by their apparent difficulty for the three-segment model from left to right, which is quantified by the value of curve-wise residual~\eqref{eq:resid_curve}. Let us refer to each curve by (date, expiration). These curves may be associated with events in the United States:
\begin{enumerate}
    \item (2024-06-05, 2024-06-12): One week before a Federal Open Market Committee meeting.
    \item (2024-11-04, 2024-11-07): The presidential election.
    \item (2024-03-05, 2025-03-12): Employment data, tariff policy, and the Consumer Price Index.
    \item (2024-06-11, 2024-06-13): Around the end of a Federal Open Market Committee meeting.
    \item (2025-04-04, 2025-04-08): Tariff policy uncertainty.
\end{enumerate}
The three rows in Figure~\ref{fig:refine}, from top to bottom, are results from a three-segment model, a four-segment model, and a hybrid version that we will specify later. A green frame indicates that the model fits the data well, whereas a red frame indicates a misfit. Let us refer to each panel by (row, column). Panel (1, 1) shows that the three-segment model can calibrate to mild W-shape curves. From panels (1, 2) to (1, 5), as the curve concavity increases, although the model does not achieve an exact fit, it still captures the overall shape.

\subsubsection*{Four-segment model}
We extend the model~\eqref{eq:qf_seqs} from $M=3$ to $M=4$ to increase expressivity. The four-segment model contains 11 parameters - two additional knot parameters and two more shape parameters in addition to those in \eqref{eq:param_seqs3}:
\begin{align} \label{eq:param_seqs4}
θ^{M=4} = θ^{M=3} ∪ \qty{u^-_3, u^+_3, γ^-_3, γ^+_3} = \qty{ς, u^-_2, u^-_3, u^+_2, u^+_3, γ^-_1, γ^-_2, γ^-_3, γ^+_1, γ^+_2, γ^+_3} .
\end{align}

The second row in Figure~\ref{fig:refine} reports the four-segment model calibration results. Panel (2, 1) is nearly identical to panel (1, 1), as the three-segment model already fits well, causing the four-segment model to collapse to it with $u^±_3 ≈ u^±_2$ and $γ^±_3 ≈ γ^±_2$. Panels (2, 2) to (2, 5) show that the four-segment model identifies additional intervals where the curve shape changes. It fits well in panels (2, 2) and (2, 3), but still slightly mismatches the large center concavity in panel (2, 4) and the large left concavity in panel (2, 5).

We stress that columns 4 and 5 represent the most extreme concavities observed in our dataset. Except for these few outliers, the four-segment model is sufficient to capture all the observed implied volatility curves.

\subsubsection*{Hybrid model}
We expect that adding more segments would eventually capture the outlier curves. As the number of parameters increases with extra segmentations, however, another model with a comparable number of parameters becomes available: an affine combination of two three-segment models.

Let $F(x; θ_i) = Q^{-1}(x; θ_i)$ for $i ∈ \qty{1,2}$ denote the two risk neutral cumulative distribution functions deduced from the two risk neutral three-segment stretched exponential quantile splices. Then the affine combination 
\begin{align} \label{eq:affine_seqs}
F_λ(x; θ_1, θ_2, λ) := λ F(x; θ_1) + (1-λ) F(x; θ_2)
\end{align}
is still a risk neutral cumulative distribution function for a suitable range of $λ ∈ \R$, where the feasibility condition will be provided in Proposition~\ref{prop:affine} in Section~\ref{sec:affine}. It contains $15$ parameters, and the three-segment model is recovered when $λ=1$.

The third row in Figure~\ref{fig:refine} confirms that the hybrid model covers all the challenging curves. Since it subsumes the three-segment model, it provides a perfect fit for all 25,770 implied volatility curves. In panels (3, 1) to (3, 5), the values of $\what{λ}$ are respectively $0.98$, $1.49$, $1.46$, $1.44$, and $1.04$.

\begin{remark} \label{rema:affine}
Comparing Figure~\ref{fig:refine} panel (1, 1) with (3, 1), we may infer that $\what{λ} < 1$ indicates the curve is within the capacity of the single three-segment model. The values $\what{λ} > 1$ for the remaining four curves, on the other hand, indicate that the hybrid model extrapolates the expressivity of the three-segment model.
\end{remark}

\subsubsection*{Alternative calibration}
\begin{figure}[h]
    \centering
    \includegraphics[width=1.0\textwidth]{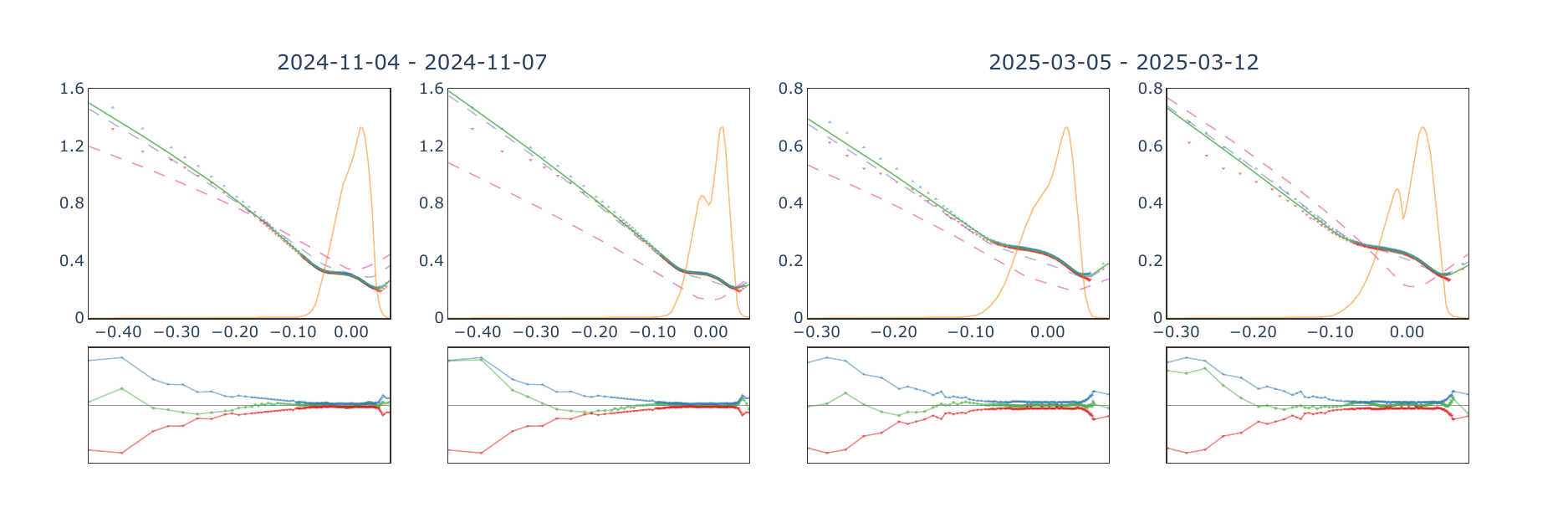}
    \caption{Alternative fits with more noticeable bimodality in implied density.}
    \label{fig:alter}
\end{figure}

We present additional results related to the flexibility of the hybrid model~\eqref{eq:affine_seqs} and the effect of the choice of weights in optimization~\eqref{eq:optim}. Figure~\ref{fig:alter} panels (1, 1) and (1, 3) reproduce the curves from Figure~\ref{fig:refine} panels (3, 2) and (3, 3). We choose the vega weight to penalize deviations from mid price less in the wing, and obtain the fits shown in Figure~\ref{fig:alter} panels (1, 2) and (1, 4). The alternative calibrations exhibit significant bimodality in the implied densities, while the fitted implied volatility curves still lie tightly within the spreads.

\section{Extrapolating expressivity} \label{sec:affine}
As demonstrated in Section~\ref{sec:refine}, the hybrid model, created by an affine combination of two three-segment stretched exponential quantile splices, demonstrates sufficient expressivity to cover all implied volatility curves in the dataset. In this section, we gain some insight into affine combinations of two risk neutral distributions.

\subsection{Feasible risk neutral affine combination}
\begin{proposition}[Feasible risk neutral affine combination] \label{prop:affine_feasible}
Let $F_1$ and $F_2$ be two risk neutral distributions with densities $f_1$ and $f_2$. Define
\begin{align}
ρ(x) = \frac{f_1(x)}{f_2(x)}, &&
\underline{ρ} := \inf_{x ∈ \R} ρ(x), &&
\overline{ρ} := \sup_{x ∈ \R} ρ(x).
\end{align}
Then
\begin{equation} \label{eq:affine}
F(x) := λF_1(x) + (1-λ) F_2(x)
\end{equation}
is also a risk neutral distribution if and only if
\begin{equation} \label{eq:affine_cond}
\frac{1}{1 - \overline{ρ}} ≤ λ ≤ \frac{1}{1-\underline{ρ}},
\end{equation}
where the case $f_1 = f_2$ reads $λ ∈ \R$.
\end{proposition}
\begin{proof}
Since $F_1$ and $F_2$ are risk neutral, risk neutrality is preserved under affine combination provided $F$ is a valid distribution. As $F$ has density $f = λf_1 + (1-λ)f_2$ with $∫f=1$, it remains to require $f ≥ 0$ and the stated condition follows.
\end{proof}

\begin{remark}
As noted in Remark~\ref{rema:affine}, the hybrid model is necessary only when $\what{λ} > 1$, in which case the combination is affine but not convex. In other words, the hybrid distribution cannot be interpreted as a compound probability distribution as in the mixture models. Instead, we shall view the hybrid model as
\begin{align}
F = F_1 - (λ - 1)(F_2 - F_1),
\end{align}
which is a functional linear extrapolation of $F_1$ in the direction away from $F_2$.
\end{remark}

The option price generated by the risk neutral affine combination is simply the corresponding affine combination of the component option prices by linearity, but the relation between the new implied volatility and the two components is more intricate. The following Lemma~\ref{prop:linear}, which follows immediately from Proposition~\ref{prop:rnd_iv}, treats more general linear combinations.

\begin{lemma}[Linear combinations of two risk neutral distributions] \label{prop:linear}
Let $F_1$ and $F_2$ be two risk neutral distributions with densities $f_1$ and $f_2$, implied volatility curves $ω_1$ and $ω_2$, and put price curves $p_1$ and $p_2$. Define the linear combinations for $λ_1, λ_2 ∈ \R$
\begin{equation}
\begin{aligned}
s(κ; λ_1, λ_2) &= λ_1 p_1(κ) + λ_2 p_2(κ), \\
G(κ; λ_1, λ_2) &= λ_1 F_1(κ) + λ_2 F_2(κ), \\
g(κ; λ_1, λ_2) &= λ_1 f_1(κ) + λ_2 f_2(κ).
\end{aligned}    
\end{equation}
If there exists a nonempty set
\begin{equation}
\what{\cK} = \qty{κ ∈ \R: ω_1(κ) = ω_2(κ)},
\end{equation}
then for $\what{κ} ∈ \what{\cK}$,
\begin{equation} \label{eq:linear}
\begin{aligned}
s(\what{κ}; λ_1, λ_2)
    &= (λ_1 + λ_2) \what{p}, \\
G(\what{κ}; λ_1, λ_2)
    &= (λ_1 + λ_2)Φ(\what{z}^+) + ϕ(\what{z}^+) \qty(λ_1\what{ω}_1' + λ_2 \what{ω}_2'), \\
g(\what{κ}; λ_1, λ_2)
        &= \frac{ϕ(\what{z}^+)}{\what{ω}} \qty(
        (λ_1 + λ_2) + \what{z}^+\what{z}^- \qty(λ_1(\what{ω}_1')^2 + λ_2(\what{ω}_2')^2 )
        - \qty(\what{z}^+ + \what{z}^-) \qty(λ_1\what{ω}_1' + λ_2 \what{ω}_2') 
        + \what{ω} \qty(λ_1\what{ω}_1'' + λ_2 \what{ω}_2'')),
\end{aligned}
\end{equation}
where
\begin{align}
\what{ω}:= ω_1(\what{κ}) = ω_2(\what{κ}), &&
\what{p}:= p_1(\what{κ}) = p_2(\what{κ}), &&
\what{z}^± := z^±(\what{κ}, \what{ω}), &&
\what{ω}_i' := ω_i'(\what{κ}), &&
\what{ω}_i'' := ω_i''(\what{κ}), &&
i ∈ \qty{1, 2}.
\end{align}
\end{lemma}

\begin{proposition}[Risk neutral affine combination] \label{prop:affine}
In the same setting as Lemma~\ref{prop:linear}, let $λ_1 = λ$ and $λ_2 = 1 - λ$ with $λ$ satisfying condition~\eqref{eq:affine_cond}. By Proposition~\ref{prop:affine_feasible}, $F(κ) = G(κ; λ, 1-λ)$ is a risk neutral distribution. Let $p$ and $ω$ be the put price curve and implied volatility curve of $F$. The following hold:
\begin{enumerate}
    \item For $κ ∈ \R$,
        \begin{align}
            p(κ) &= λ p_1(κ) + (1-λ) p_2(κ), \\
            e^κ Φ(z^+) - Φ(z^-) &= e^κ \qty(λ Φ(z^+_1) + (1-λ) Φ(z^+_2)) - \qty(λ Φ(z^-_1) + (1-λ) Φ(z^-_2)),
        \end{align}
        and when $κ=0$, 
        \begin{equation}
            ω(0) = 2 Φ^{-1} \qty(λΦ\qty(\frac{ω_1(0)}{2}) + (1-λ)Φ\qty(\frac{ω_2(0)}{2})),
        \end{equation}
        where
        \begin{align}
        z^±_i := z^±(κ, ω_i(κ)), &&
        ω_i := ω_i(κ), &&
        ω_i' := \frac{\dd ω_i(κ)}{\dd κ}, &&
        ω_i'' := \frac{\dd^2 ω_i(κ)}{\dd κ^2}, &&
        i ∈ \qty{1, 2}.
        \end{align}
    \item For $\what{κ} ∈ \what{\cK}$,
        \begin{align}
         ω(\what{κ})
            &= \what{ω}, \\
         ω'(\what{κ})
            &= λ\what{ω}_1' + (1-λ) \what{ω}_2', \\
         ω''(\what{κ})
            &= λ \what{ω}_1'' + (1-λ) \what{ω}_2'' + \frac{λ (1-λ)\what{z}^+\what{z}^- }{\what{ω}} \qty(\what{ω}_1' - \what{ω}_2')^2,
        \end{align}
        where
        \begin{align}
        \what{z}^± := z^±(\what{κ}, \what{ω}), &&
        \what{ω}:= ω_1(\what{κ}) = ω_2(\what{κ}), &&
        \what{ω}_i' := ω_i'(\what{κ}), &&
        \what{ω}_i'' := ω_i''(\what{κ}), &&
        i ∈ \qty{1, 2}.
        \end{align}
\end{enumerate}
\end{proposition}

\begin{proof}
For part 1, the first equation follows from $∫_{-∞}^κ (e^κ - e^x) \dd F(x) = ∫_{-∞}^κ (e^κ - e^x) \dd (λ F_1(x) + (1-λ) F_2(x))$. Substituting the pricing formula~\eqref{eq:price_bs} yields the second result. The third result follows directly from the second equation. For part 2, as $p(\what{κ}) = s(\what{κ}; λ, 1-λ)$ and the put price is monotone with respect to implied volatility, we have $ω(\what{κ}) = \what{ω}$. We then apply Proposition~\ref{prop:rnd_iv} and Lemma~\ref{prop:linear}, and compare $F(\what{κ})$ and $f(\what{κ})$ in equation~\eqref{eq:rnd_iv} with $G(\what{κ}; λ, 1-λ)$ and $g(\what{κ}; λ, 1-λ)$ in equation~\eqref{eq:linear}, respectively. The result follows.
\end{proof}

As shown in the first part of Proposition~\ref{prop:affine}, the relation between the hybrid implied volatility $ω$ and the two component implied volatilities $ω_1$ and $ω_2$ is implicit except at the money $κ=0$, and it is even more intricate to infer the slope and convexity of the hybrid implied $ω$. The second part, nonetheless, shows that the value, slope, and convexity of $ω$ are easier to infer at the intersection, if any, of $ω_1$ and $ω_2$. This can be examined in the third row of Figure~\ref{fig:refine} and in the first row of Figure~\ref{fig:alter}.

\subsection{Contrasting two risk neutral distributions}
\begin{proposition} [Contrasting two risk neutral distributions] \label{prop:contrast}
In the same setting as Lemma~\ref{prop:linear}, define $ΔF := F_1 - F_2$, $Δf := f_1 - f_2$, $Δω := ω_1 - ω_2$, $Δp := p_1 - p_2$. The following statements hold:
\begin{enumerate}
    \item For every $κ ∈ \R$
        \begin{equation}
            \sign Δp(κ) = \sign Δω(κ).
        \end{equation}
        In particular, if $\what{κ} ∈ \what{\cK}$, then
        \begin{align}
        \sign ΔF(\what{κ}) = \sign Δω'(\what{κ}).
        \end{align}
        Additionally, if $\what{κ}$ is a critical point of $Δω$, then
        \begin{align}
        \sign Δf(\what{κ}) = \sign Δω''(\what{κ}).
        \end{align}
    \item Define the number of strict zero crossings of $h$ by
        \begin{equation}
            χ(h) = \# \qty{κ: h(κ)=0, \lim_{ε → 0} h(κ-ε)h(κ+ε) <0},
        \end{equation}
        then
        \begin{align}
            χ(Δp) = χ(Δω) , && χ(ΔF) ≥ χ(Δω) + 1, && χ(Δf) ≥ χ(Δω) + 2.
        \end{align}
\end{enumerate}
\end{proposition}

\begin{proof}
For part 1, the first result follows from strict monotonicity of $p$ with respect to $ω$. At $\what{κ} ∈ \what{\cK}$, we calculate $ΔF(\what{κ}) = G(\what{κ}; 1, -1) = ϕ(\what{z}^+) Δ\what{ω}'$, whose sign is determined by the second factor. With $\what{ω}'=0$, it holds that $Δf(\what{κ}) = g(\what{κ}; 1, -1) = ϕ(\what{z}^+) Δ\what{ω}''$ and thus its sign is determined by the second factor. 

For part 2, by part 1 we have $χ(Δp) = χ(Δω)$. Let $n = χ(Δω)$, and denote the points at which $Δω$ crosses zero by $\what{κ}_1 < \what{κ}_2 < … < \what{κ}_n$. In addition, $Δp(-∞) = Δp(+∞)=0$, so using the conventions $\what{κ}_0 := - ∞$, $\what{κ}_{n+1} := + ∞$ we have $Δp(\what{κ})= 0$ for $\what{κ} ∈ \qty{\what{κ}_i}_{i=0}^{n+1}$. By Rolle's theorem, for each interval $i = 1, \dots, n+1$ there is at least one $x_i ∈ (\what{κ}_{i-1}, \what{κ}_i)$ such that $Δp'(x_i) = e^{x_i} ΔF(x_i) = 0$, and hence $ΔF(x_i) = 0$. Since $Δp$ crosses zero strictly at $\qty{\what{κ}_i}_{i=1}^n$, its sign alternates through intervals $(\what{κ}_{i-1}, \what{κ}_i)$ for $i=1, \dots, n+1$, which determines the direction in which $ΔF$ crosses zero at $x_i$. A positive value of $Δp$ in $(\what{κ}_{i-1}, \what{κ}_i)$ implies $ΔF(x_i)$ crosses zero from above, and vice versa. Therefore, $χ(ΔF) ≥ χ(Δp) + 1 = n + 1$. Finally, because $ΔF(-∞)=ΔF(+∞)=0$ and $Δf = ΔF'$, using similar arguments we obtain $χ(Δf) ≥ χ(ΔF) + 1 ≥ n+2$.
\end{proof}

\begin{remark}
Proposition~\ref{prop:contrast} slightly generalizes \citep[Proposition 2.1]{glasserman2023w}. It allows one to contrast any two risk neutral distributions. Thus the density crossover inference can be applied to a tilted W-shape - for example, Figure~\ref{fig:refine} (1, 2) - which requires a nonhorizontal line to obtain four intersections.
\end{remark}

\section{Dynamic stretched exponential quantile splice} \label{sec:dynamic}
The precise calibration performance over the dataset sets a solid foundation for the two downstream tasks: reconstructing the entire implied volatility surface by interpolating the individual curves without arbitrage, and constructing a dynamic process that reproduces the surface.

For the first task, as discussed in Section~\ref{sec:param_seqsterm}, the fitted parameters~\eqref{eq:param_seqs3t}, shown in Figure~\ref{fig:demo_param}, exhibit a clear pattern across tenor, indicating the existence of a feasible term structure. Hence, we let the parameter~\eqref{eq:param_seqs} be a vector valued function of tenor as
\begin{align} \label{eq:param_seqst}
θ(τ) = \qty{ς(τ)} ∪ \qty{u^±_m(τ)}_{m=1}^{M-1} ∪ \qty{γ^±_m(τ)}_{m=0}^{M-1}: \cT → \R_+ × \cU^{M-1} × \R_+^{M}.
\end{align}
If the term structure rules out arbitrage, then the second task can be completed with a Markov diffusion with the \citep{dupire1994pricing} local volatility. In this section, we simply present some general results.

\subsection{Feasible term structure}
Let $Q(u; θ)$ and $F(x; θ)$ denote, respectively, the risk neutral quantile and its corresponding distribution parameterized by $θ ∈ \R^d$. Assigning a term structure to the parameters as $θ(τ): \cT → \R^d$ defines the dynamic quantile function and cumulative distribution function, together with their corresponding density functions
\begin{align} \label{eq:qf_dynamic}
Q(τ, u):= Q(u; θ(τ)), && q(τ, u):= \frac{\dd}{\dd u} Q(τ, u), \\
F(τ, x):= F(x; θ(τ)), && f(τ, x):= \frac{\dd}{\dd x} F(τ, x).
\end{align}
It is useful to express one in terms of the other via
\begin{align} \label{eq:qf_cdf_convert}
F(τ, Q(τ, u)) = u, &&
f(τ, Q(τ, u)) = \frac{1}{q(τ, u)}, &&
∂_τ F(τ, Q(τ, u)) = - \frac{∂_τ Q(τ, u)}{q(τ, u)}.
\end{align}
For the process of logarithmic returns,
\begin{align}
 X_τ ∼ F(τ, ·) = F(·; θ(τ)), && X_0 = 0,
\end{align}
the corresponding options can be priced by substituting the term structure $θ(τ)$ into equation~\eqref{eq:price}:
\begin{align}
p(τ, κ) = p(κ; θ(τ)) = e^κ F(κ; θ(τ)) - \wtilde{F}(κ; θ(τ)) = e^κ F(κ; θ(τ)) - ∫_0^{F(κ; θ(τ))} e^{Q(u; θ(τ))} \dd u.
\end{align}
$p$ is free of static arbitrage if for all $(τ, κ) ∈ \cT × \cK$
\begin{align} \label{eq:arb_cond}
∂_τ p(τ, κ) ≥ 0, &&
∂_κ p(τ, κ) ≥ 0, &&
∂_{κκ} p(τ, κ) ≥ 0,
\end{align}
which corresponds to the conditions excluding calendar spread, vertical spread, and butterfly spread arbitrage. By construction, $F(x; θ(τ))$ is risk neutral, so the last two types of arbitrage in the moneyness dimension are automatically ruled out. The first type of arbitrage, along the tenor dimension, must be ruled out through the parameter term structure $θ(τ)$.

\begin{proposition}[Risk neutral term structure] \label{prop:term}
A feasible term structure $θ(τ)$ satisfies:
\begin{enumerate}
    \item Initial condition:
        \begin{align}
        Q(u; θ(0)) = 0, &&
        \text{or} &&
        F(x; θ(0)) = δ(x).
        \end{align}
    \item No arbitrage condition: for all $τ ∈ \cT$, $u ∈ \cU$, and $x ∈ \cK$
        \begin{align}
        ∫_0^{u} e^{Q(v; θ(τ))} \nabla_θ Q(v; θ(τ)) · θ'(τ) \dd v ≤ 0, &&
        \text{or} &&
        ∫_{-∞}^x e^y \nabla_θ F(y; θ(τ)) · θ'(τ) \dd y  ≥ 0.
        \end{align}
\end{enumerate}
\end{proposition}
\begin{proof}
The initial condition follows by construction. The no arbitrage condition in terms of $F$ follows from
\begin{align}
∂_τ p(τ, κ)
    &= e^κ ∂_τ F(τ, κ) - ∫_{-∞}^κ e^x ∂_τ f(τ, x) \dd x 
    = ∫_{-∞}^κ e^x ∂_τ F(τ, x) \dd x, \\
    &= ∫_{-∞}^κ e^x ∂_τ F(x; θ(τ)) \dd x
    = ∫_{-∞}^κ e^x \nabla_θ F(x; θ(τ)) · θ'(τ) \dd x ≥ 0.
\end{align}
Applying equation~\eqref{eq:qf_cdf_convert} and a change of variables yields the no arbitrage condition in terms of $Q$.
\end{proof}

By combining Propositions~\ref{prop:neutr} and \ref{prop:term}, one specifies a family of marginal distributions that are free of arbitrage. A local volatility model can be constructed to match the distribution.

\subsection{Local volatility with term structure}
Consider the stochastic differential equation
\begin{align} \label{eq:sde}
\dd X_τ = - \frac{1}{2} σ^2(τ, X_τ) \dd τ + σ(τ, X_τ) \dd W_τ,
&&
X_0 = 0
\end{align}
that has marginal distribution $F(τ, ·)$, quantile function $Q(τ, ·)$, and put price $p(τ, ·)$, where the local volatility $σ: \cT × \cK → \R_+$ is expressed in terms of tenor and moneyness.

\begin{proposition}[Local volatility in moneyness]
The local volatility of the stochastic differential equation~\eqref{eq:sde} can equivalently be expressed in terms of
\begin{enumerate}
    \item normalized put option:
        \begin{align} \label{eq:lv_put}
            σ^2(τ, κ) = \frac{2 ∂_τ p(τ, κ)}{∂_{κ κ} p(τ, κ) - ∂_κ p(τ, κ)},
        \end{align}
    \item cumulative distribution function:
        \begin{align} \label{eq:lv_cdf}
            σ^2(τ, x) =\frac{2}{e^x f(τ, x)} ∫_{-∞}^x e^y ∂_τ F(τ, y) \dd y,
        \end{align}
    \item quantile function:
        \begin{align} \label{eq:lv_qf}
            σ^2(τ, Q(τ, u)) = - \frac{2q(τ, u)}{e^{Q(τ, u)}} ∫_0^u e^{Q(τ, v)} ∂_τ Q(τ, v) \dd v.
        \end{align}
\end{enumerate}
\end{proposition}

\begin{proof}
Write the corresponding Fokker-Planck equation and integrate by parts twice:
\begin{align}
∂_τ f
    &= \frac{1}{2} ∂_x σ^2 f + \frac{1}{2} ∂_{xx} σ^2 f \\
∂_τ F
    &= \frac{1}{2}  σ^2 f + \frac{1}{2} ∂_x σ^2 f \\
e^x ∂_τ F
    &= \frac{1}{2} ∂_x (e^x σ^2 f).
\end{align}
Integrating over moneyness yields equation~\eqref{eq:lv_cdf}. Equation~\eqref{eq:lv_put} follows by substituting
\begin{align}
∂_τ p(τ, κ) = ∫_{-∞}^κ e^x ∂_τ F(τ, x) \dd x, &&
∂_κ p(τ, κ) = e^κ F(τ, κ), &&
∂_{κ κ} p(τ, κ) = e^κ F(τ, κ) + e^κ f(τ, κ).
\end{align}
Finally, equation~\eqref{eq:lv_qf} follows by applying equation~\eqref{eq:qf_cdf_convert} to equation~\eqref{eq:lv_cdf}.
\end{proof}

\begin{corollary}[Local volatility with term structure]
If $F$ and $Q$ can be expressed with term structure $θ(τ)$ as in Proposition~\ref{prop:term}, then
\begin{align}
σ^2(τ, x) = σ^2(x; θ(τ))
    &= \frac{2 ∂_τ p(κ; θ(τ))}{∂_{κ κ} p(κ; θ(τ)) - ∂_κ p(κ; θ(τ))} \\
    &= \frac{2}{e^x f(x; θ(τ))} ∫_{-∞}^x e^y \nabla_θ F(y; θ(τ)) · θ'(τ) \dd y \\
    &= - \frac{2q(u_x; θ(τ))}{e^{Q(u_x; θ(τ))}} ∫_0^{u_x} e^{Q(v; θ(τ))} \nabla_θ Q(v; θ(τ)) · θ'(τ) \dd v,
\end{align}
where $u_x := Q^{-1}(x; θ(τ))$.
\end{corollary}

\setcitestyle{numbers}
\bibliographystyle{chicago}
\bibliography{ref}

@book{gilchrist2000statistical,
  title={Statistical modelling with quantile functions},
  author={Gilchrist, Warren},
  year={2000},
  publisher={Chapman and Hall/CRC}
}

@article{dupire1994pricing,
  title={Pricing with a smile},
  author={Dupire, Bruno},
  journal={Risk},
  volume={7},
  number={1},
  pages={18--20},
  year={1994}
}

@article{carr2021additive,
  title={Additive logistic processes in option pricing},
  author={Carr, Peter and Torricelli, Lorenzo},
  journal={Finance and Stochastics},
  volume={25},
  number={4},
  pages={689--724},
  year={2021},
  publisher={Springer}
}

@article{ackerer2020deep,
  title={Deep smoothing of the implied volatility surface},
  author={Ackerer, Damien and Tagasovska, Natasa and Vatter, Thibault},
  journal={Advances in Neural Information Processing Systems},
  volume={33},
  pages={11552--11563},
  year={2020}
}

@article{wiedemann2024operator,
  title={Operator deep smoothing for implied volatility},
  author={Wiedemann, Ruben and Jacquier, Antoine and Gonon, Lukas},
  journal={arXiv 2406.11520},
  year={2024}
}

@article{durrleman2010implied,
  title={Implied volatility: Market models},
  author={Durrleman, Valdo},
  journal={Encyclopedia of Quantitative Finance},
  year={2010},
  publisher={Wiley Online Library}
}

@phdthesis{roper2009implied,
  title={Implied volatility: General properties and asymptotics},
  author={Roper, Michael},
  year={2009},
  school={UNSW Sydney}
}

@article{tavin2012implied,
  title={Implied distribution as a function of the volatility smile},
  author={Tavin, Bertrand},
  journal={Bankers Markets and Investors},
  number={119},
  pages={31--42},
  year={2012}
}

@article{brunner2003arbitrage,
  title={Arbitrage-free estimation of the risk-neutral density from the implied volatility smile},
  author={Brunner, Bernhard and Hafner, Reinhold},
  journal={Journal of Computational Finance},
  year={2003},
  volume={7},
  number={1},
  pages={75--106}
}

@article{jackwerth2000recovering,
  title={Recovering risk aversion from option prices and realized returns},
  author={Jackwerth, Jens Carsten},
  journal={The Review of Financial Studies},
  volume={13},
  number={2},
  pages={433--451},
  year={2000},
  publisher={Oxford University Press}
}

@inproceedings{zheng2021incorporating,
  title={Incorporating prior financial domain knowledge into neural networks for implied volatility surface prediction},
  author={Zheng, Yu and Yang, Yongxin and Chen, Bowei},
  booktitle={Proceedings of the 27th ACM SIGKDD Conference on Knowledge Discovery \& Data Mining},
  pages={3968--3975},
  year={2021}
}

@inproceedings{yang2025hyperiv,
  title={HyperIV: Real-time implied volatility smoothing},
  author={Yang, Yongxin and Chen, Wenqi and Shu, Chao and Hospedales, Timothy},
  booktitle={The 42nd International Conference on Machine Learning},
  pages={1--15},
  year={2025},
  organization={PMLR}
}

@article{azzone2025explicit,
  title={Explicit option pricing with additive processes},
  author={Azzone, Michele and Torricelli, Lorenzo},
  journal={SIAM Journal on Financial Mathematics},
  volume={16},
  number={3},
  pages={747--802},
  year={2025},
  publisher={SIAM}
}

@inproceedings{lin2024neural,
  title={Neural Term Structure of Additive Process for Option Pricing},
  author={Lin, Jimin and Liu, Guixin},
  booktitle={Proceedings of the 5th ACM International Conference on AI in Finance},
  pages={695--702},
  year={2024}
}

@article{lin2026shallow,
  title={Shallow Representation of Option Implied Information},
  author={Lin, Jimin},
  journal={arXiv 2603.17151},
  year={2026}
}

@article{lee2004moment,
  title={The moment formula for implied volatility at extreme strikes},
  author={Lee, Roger W},
  journal={Mathematical Finance},
  volume={14},
  number={3},
  pages={469--480},
  year={2004},
  publisher={Wiley Online Library}
}

@article{benaim2009regular,
  title={Regular variation and smile asymptotics},
  author={Benaim, Shalom and Friz, Peter},
  journal={Mathematical Finance},
  volume={19},
  number={1},
  pages={1--12},
  year={2009},
  publisher={Wiley Online Library}
}

@article{brigo2002lognormal,
  title={Lognormal-mixture dynamics and calibration to market volatility smiles},
  author={Brigo, Damiano and Mercurio, Fabio},
  journal={International Journal of Theoretical and Applied Finance},
  volume={5},
  number={04},
  pages={427--446},
  year={2002},
  publisher={World Scientific}
}

@article{glasserman2023w,
  title={W-shaped implied volatility curves and the {G}aussian mixture model},
  author={Glasserman, Paul and Pirjol, Dan},
  journal={Quantitative Finance},
  volume={23},
  number={4},
  pages={557--577},
  year={2023},
  publisher={Taylor \& Francis}
}

@article{zaugg2025volatility,
  title={Volatility parametrizations with random coefficients: Analytic flexibility for implied volatility surfaces},
  author={Zaugg, Nicola F and Perotti, Leonardo and Grzelak, Lech A},
  journal={SIAM Journal on Financial Mathematics},
  volume={16},
  number={4},
  pages={1374--1411},
  year={2025},
  publisher={SIAM}
}

@article{keller2025w,
  title={W-shaped implied volatility curves in a variance-gamma mixture model},
  author={Keller-Ressel, Martin},
  journal={Frontiers of Mathematical Finance},
  volume={6},
  pages={1--15},
  year={2025},
  publisher={Frontiers of Mathematical Finance}
}

@article{itkin2015sigmoid,
  title={To sigmoid-based functional description of the volatility smile},
  author={Itkin, Andrey},
  journal={The North American Journal of Economics and Finance},
  volume={31},
  pages={264--291},
  year={2015},
  publisher={Elsevier}
}

@article{alexiou2025pricing,
  title={Pricing event risk: Evidence from concave implied volatility curves},
  author={Alexiou, Lykourgos and Goyal, Amit and Kostakis, Alexandros and Rompolis, Leonidas},
  journal={Review of Finance},
  volume={29},
  number={4},
  pages={963--1007},
  year={2025},
  publisher={Oxford University Press}
}

@article{baker2018trading,
  title={Trading events: Smiles, frowns and moustaches--the many faces of the options market},
  author={Baker, Mark and Gillberg, Tor and Thomas, Shaun},
  journal={SSRN 3263368},
  year={2018}
}

@misc{klassen2023state,
  title={State of the Smile: The ever-surprising evolution of the equity/listed options market},
  author={Klassen, Timothy},
  note={VolaDynamics},
  year={2023}
}

@misc{zhou20250dte,
  title={0DTE smile dynamics and jump diffusion models},
  author={Zhou, Shengquan},
  note={Lecture slides, NYU Tandon},
  year={2025},
}

@article{keller2026discovering,
  title={Discovering parametrizations of implied volatility with symbolic regression},
  author={Keller-Ressel, Martin and Nikulski, Hannes},
  journal={arXiv 2603.21892},
  year={2026}
}

@article{keelin2016metalog,
  title={The metalog distributions},
  author={Keelin, Thomas W},
  journal={Decision Analysis},
  volume={13},
  number={4},
  pages={243--277},
  year={2016},
  publisher={INFORMS}
}

@article{homescu2011implied,
  title={Implied volatility surface: Construction methodologies and characteristics},
  author={Homescu, Cristian},
  journal={arXiv 1107.1834},
  year={2011}
}

\appendix
\section{Additional results} \label{sec:appendix}
\begin{figure}
    \centering
    \includegraphics[width=1.0\linewidth]{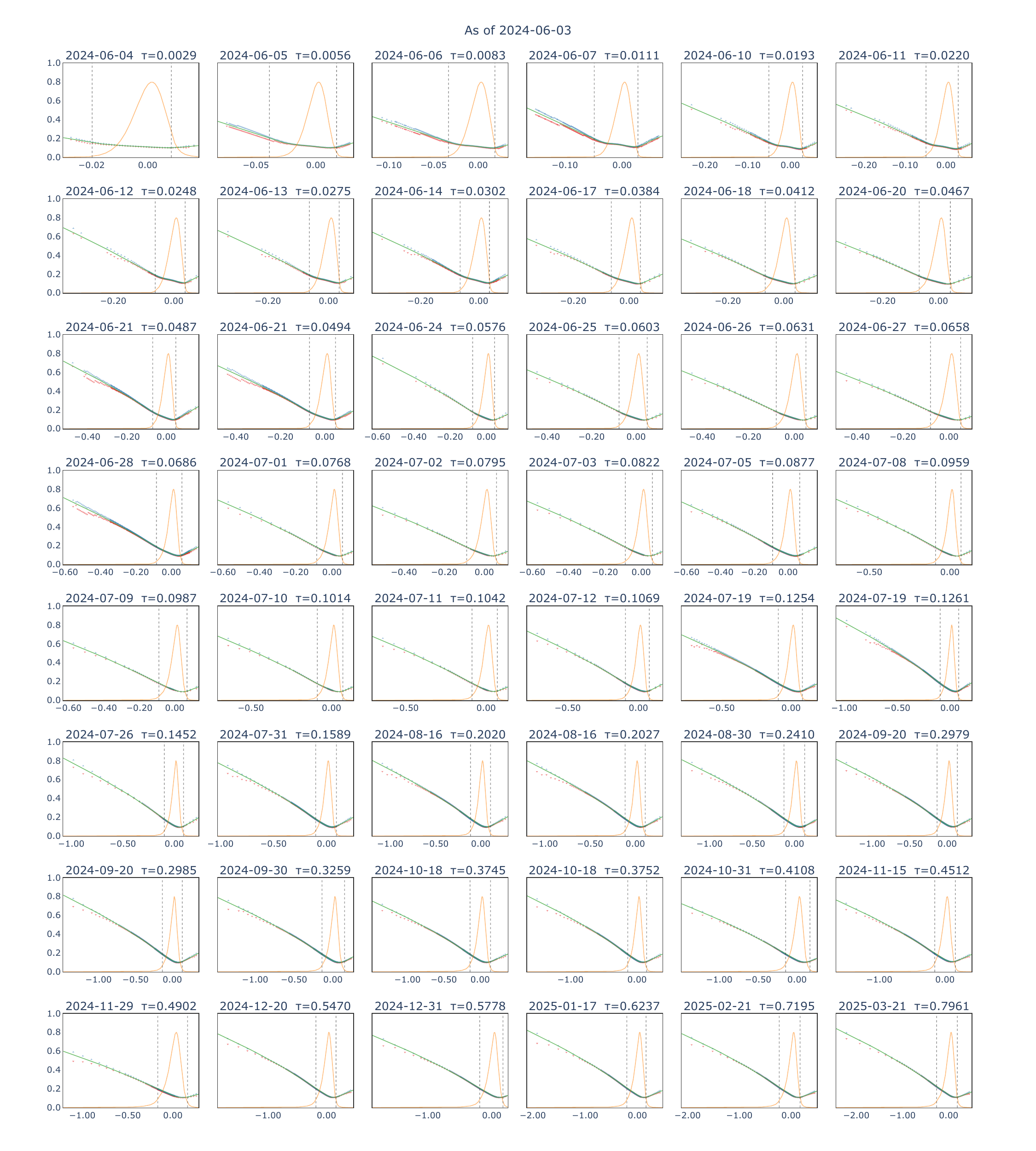}
    \caption{Calibration for 2024-06-03.}
\end{figure}

\begin{figure}
    \centering
    \includegraphics[width=1.0\linewidth]{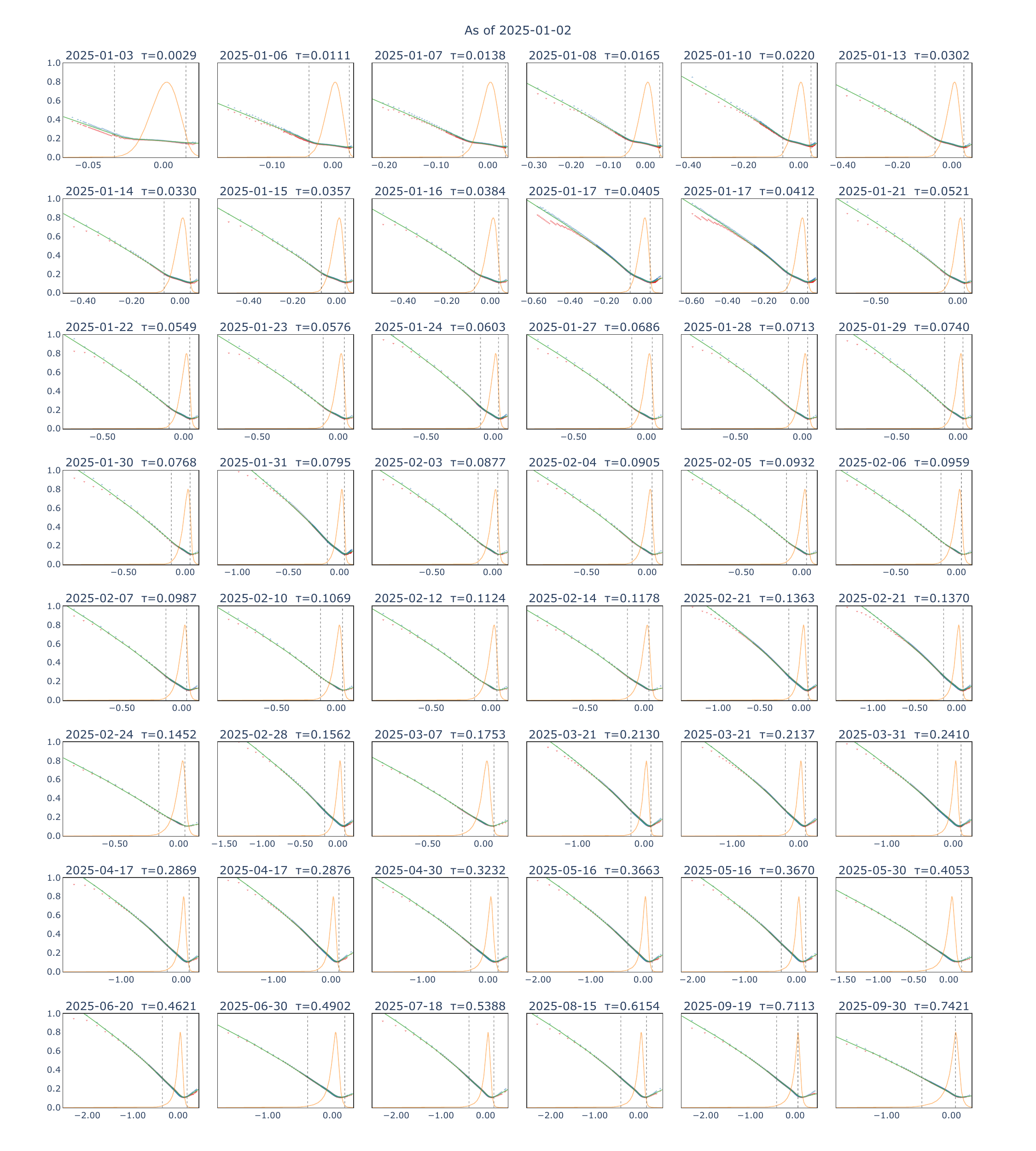}
    \caption{Calibration for 2025-01-02.}
\end{figure}

\begin{figure}
    \centering
    \includegraphics[width=1.0\linewidth]{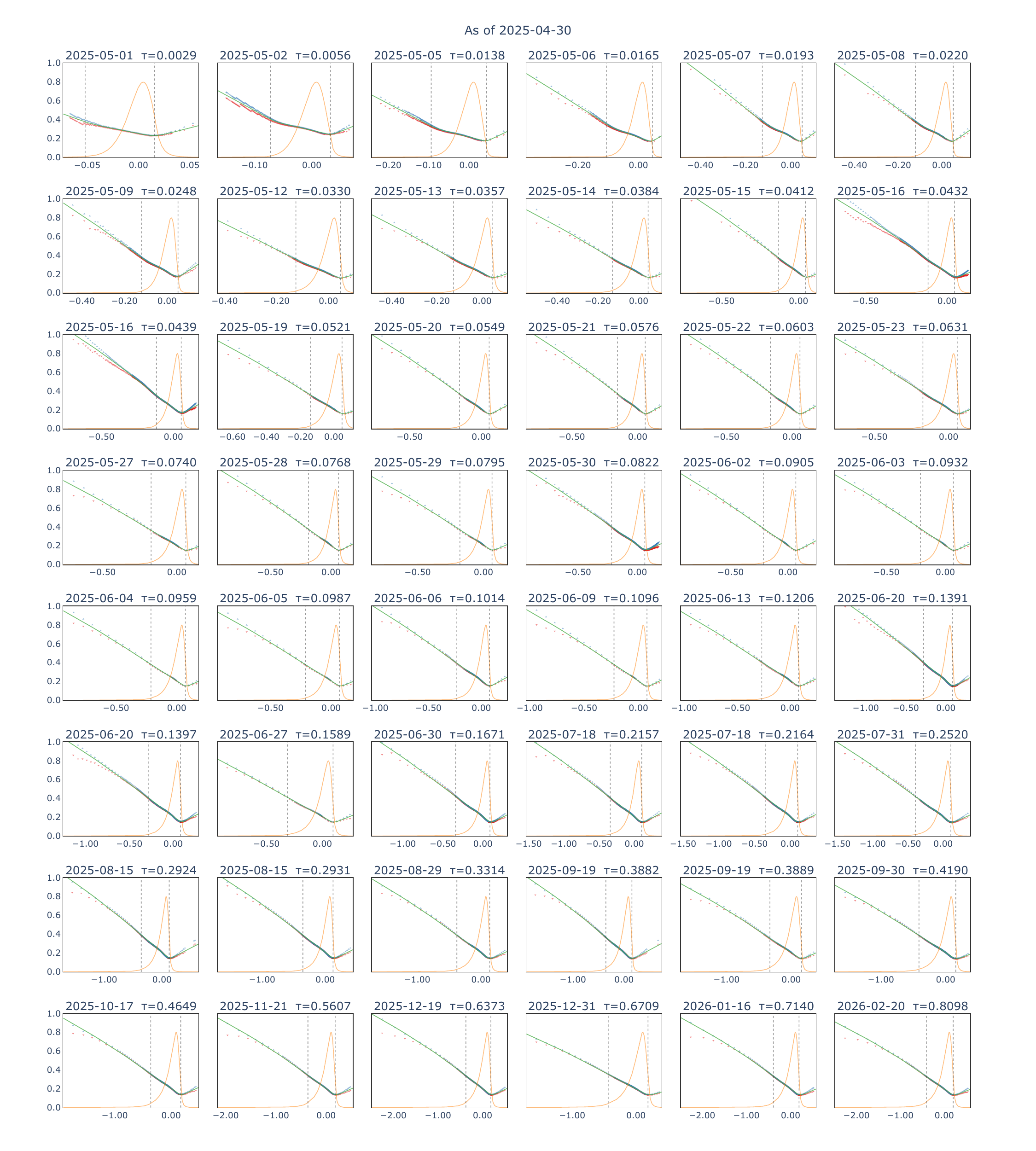}
    \caption{Calibration for 2025-04-30.}
\end{figure}

\begin{figure}[h]
    \centering
    \includegraphics[width=1.0\linewidth]{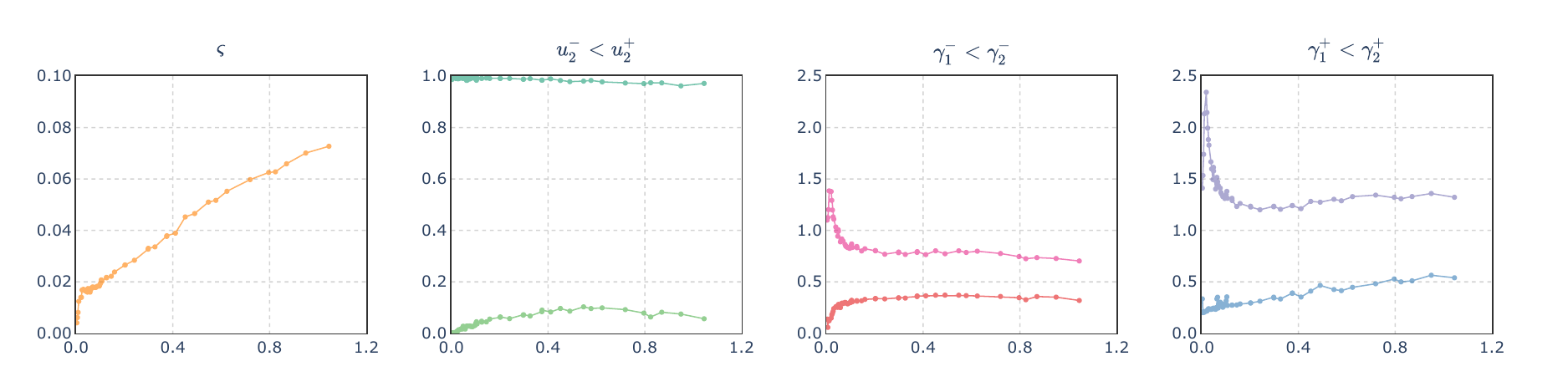}
    \caption{Fitted parameters for 2024-06-03.}
\end{figure}

\begin{figure}[h]
    \centering
    \includegraphics[width=1.0\linewidth]{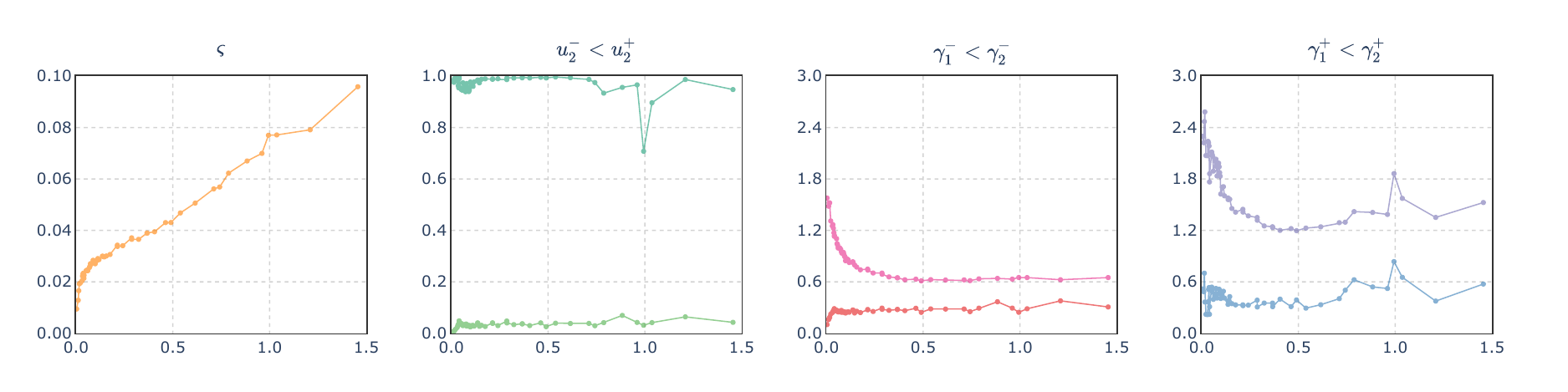}
    \caption{Fitted parameters for 2025-01-02.}
\end{figure}

\begin{figure}[h]
    \centering
    \includegraphics[width=1.0\linewidth]{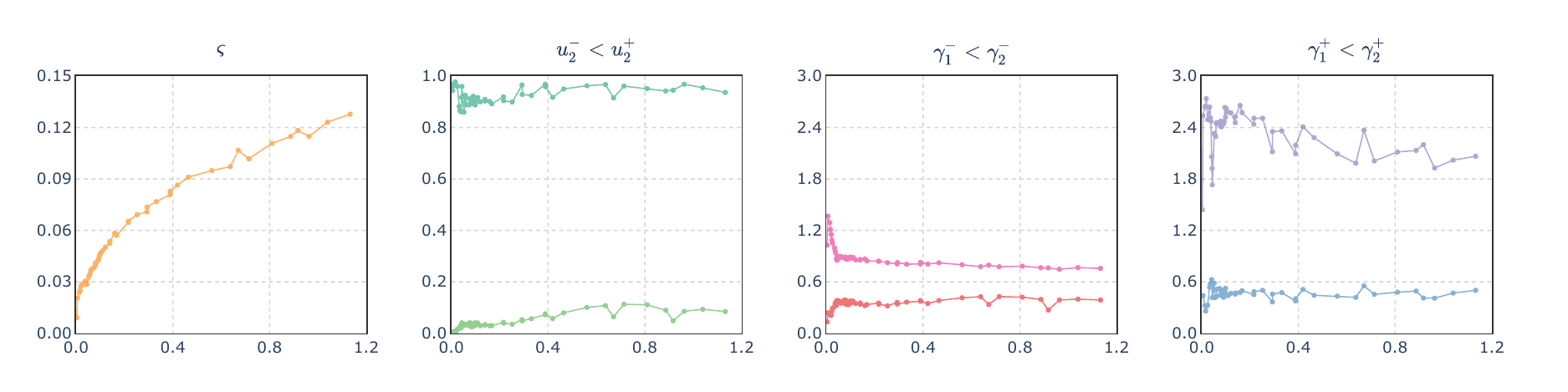}
    \caption{Fitted parameters for 2025-04-30.}
\end{figure}
\end{document}